\documentclass[11pt]{article}

\usepackage[margin=1in]{geometry}
\usepackage{amsmath,amssymb,amsthm,mathtools}
\usepackage{bm}
\usepackage[ruled,vlined]{algorithm2e}
\usepackage{booktabs}
\usepackage{graphicx}
\usepackage{enumitem}
\usepackage{microtype}
\usepackage[colorlinks=true,linkcolor=blue,citecolor=blue,urlcolor=blue]{hyperref}
\hypersetup{pdftitle={Dynamical low-rank equilibrium computation for stochastic games between advanced persistent threats and moving target defense},pdfauthor={Zijian Tian, He Zhang, Xinjie Chen, Wenhai Wang, Xinggao Liu},pdfkeywords={advanced persistent threats, moving target defense, zero-sum stochastic games, dynamical low-rank approximation, equilibrium computation, industrial control systems}}
\usepackage{natbib}
\usepackage{float}
\usepackage[font=small]{caption}

\theoremstyle{plain}
\newtheorem{theorem}{Theorem}
\newtheorem{lemma}{Lemma}
\newtheorem{corollary}{Corollary}
\newtheorem{proposition}{Proposition}
\theoremstyle{definition}

\newtheorem{assumption}{Assumption}
\newtheorem{remark}{Remark}
\newtheorem{problem}{Problem}
\theoremstyle{remark}

\title{Dynamical low-rank equilibrium computation for stochastic games between advanced persistent threats and moving target defense\thanks{The source code is available at \url{https://github.com/tz98lab/dlr-ne-ics-security}.}}
\author{Zijian Tian$^{1}$, He Zhang$^{1}$, Xinjie Chen$^{1,2}$, Wenhai Wang$^{1}$, and Xinggao Liu$^{1,}$\thanks{Corresponding author: Xinggao Liu (email: lxg@zju.edu.cn).}\\
\textit{$^{1}$College of Control Science and Engineering, Zhejiang University, PC 310027, China}\\
\textit{$^{2}$Alibaba Group, Hangzhou, China}}
\date{}

\begin{document}
\maketitle

% ============================================================
% ABSTRACT (rewrite v2)
% ============================================================
\begin{abstract}
Moving target defense (MTD) against advanced persistent threats (APTs) in industrial control systems (ICS) has well-established game-theoretic formulations, but their practical value hinges on equilibrium computation, which faces two gaps: full-rank value iteration is prohibitively expensive at industrial scale, and the resulting defense strategies admit no certified robustness against adversarial perturbations.
We first reveal that the attack and defense influence matrices of ICS dynamics are intrinsically low-rank: APTs infiltrate through a handful of entry points and MTD reconfigures only a few components per cycle.
Our theory makes four contributions.
First, an augmented gradient matrix certifies that the low-rank structure propagates through the non-smooth Bellman operator of the zero-sum stochastic game, so that every Bellman target lies near a low-dimensional subspace and low-rank truncation incurs an explicit error bound (Lemma~1, Theorem~1).
Second, we propose the \underline{D}ynamical \underline{L}ow-\underline{R}ank \underline{N}ash \underline{E}quilibrium algorithm, named DLR-NE, which augments the rank-$r$ search space each iteration, regularizes the core matrix spectrum, and retracts via truncated SVD, and prove that it converges geometrically to a neighborhood whose error decomposes into five physically interpretable sources (Theorem~2).
Third, its per-step cost is $\mathcal{O}(nr^2)$, a $\Theta(n/r^2)$ speedup over full-rank value iteration (Theorem~3).
Fourth, a single weight trades accuracy against a certified sensitivity bound of the induced defense strategy under core-matrix perturbations (Corollary~1).
Six experiments on a nonlinear power-system testbed confirm each prediction, with $94\%$ parameter compression at $2.3\%$ utility loss.
We make all of our experimental data and code publicly available.$^{*}$
\end{abstract}


% ============================================================
% SECTION 1: INTRODUCTION
% ============================================================
% ============================================================
% SECTION 1: INTRODUCTION (rewrite v2, ~960 words, target ~1.5 pages)
% ============================================================
\section{Introduction}

Advanced persistent threats (APTs) are prolonged, stealthy campaigns that infiltrate industrial control systems (ICS), evade detection, and progressively seize control of critical components, constituting one of the most formidable challenges to industrial cyber-physical security~\citep{alshamrani2019aptsurvey,xu2026sr}. Static defenses---fixed configurations, signature-based monitoring, physical isolation---are increasingly ineffective against such evolving campaigns, and learning-based intrusion detection struggles to keep pace with adversaries that continuously adapt to evade classification~\citep{lv2020ids,shen2026stapt,presekal2025aptpower}.

Moving target defense (MTD) has therefore emerged as a dynamic paradigm that periodically reconfigures system configurations to invalidate adversarial reconnaissance~\citep{cho2020mtdsurvey,zhang2025mtdaisurvey}. Most existing MTD strategies, however, follow offline periodic schedules: once the adversary discerns the rhythm, the attack surface becomes predictable again. Adaptive scheduling and multi-layer coordination~\citep{umsonst2024bayesmtd,zang2026aptmtd,nespoli2026reactive} improve responsiveness but remain confined to discrete event models or static optimization. Stochastic game theory addresses this gap by encoding state evolution, action coupling, and long-term objectives in a unified framework~\citep{shapley1953stochastic,lye2005game}, and attacker--defender interactions have accordingly been modeled as zero-sum or dynamic games~\citep{bai2017datainjection,miao2018hybrid,ugrinovskii2017jammer,wu2020zerosum,li2025encryption,huang2020dynamic,jiao2026bgflipin,qin2026insider}. These formulations, however, remain largely confined to finite discrete states or small-scale graphs that do not capture the high-dimensional continuous dynamics of industrial control systems.

However, industrial reality is governed by nonlinear dynamics in continuous time, where voltages, pressures, and temperatures constitute thousands of coupled variables; APT trajectories evolve continuously, and defensive reconfigurations cannot await offline batch processing. Equilibrium computation in this regime confronts two barriers. First, full-rank value iteration scales poorly with the state dimension: for a typical system with $n$ ranging from $10^3$ to $10^4$, a single-hidden-layer network entails $10^6$ to $10^8$ parameters and each Bellman update demands $10^9$ to $10^{12}$ floating-point operations, far exceeding real-time control budgets. Low-rank factorizations are a mainstream remedy in parameter-efficient learning~\citep{hu2022lora}, yet for game-theoretic equilibrium computation neither their approximation accuracy nor their robustness implications are understood. Second, the resulting equilibrium policies lack certified robustness: unlike classical $H_\infty$ controllers whose degradation under uncertainty is bounded a priori~\citep{zhu2011robust}, neural parameterizations admit no computable sensitivity bound, so adversarial perturbations of weights may induce arbitrarily large deviations of the induced defense strategy.

To resolve both barriers, we inspect the physical structure of industrial dynamics. The high-dimensional process state is influenced by attacker and defender actions only through a small number of independent channels: APTs infiltrate through a handful of entry points and manipulate a limited set of actuators or links, while MTD can reconfigure only a restricted subset of relays, segments, or control loops per cycle owing to resource and latency constraints. The attack and defense influence matrices are therefore intrinsically low-rank at every operating point. Full-rank parameterizations ignore this embedded structure, incurring both the computational curse of dimensionality and a structural fragility without quantifiable safety margin.

\textbf{Our Contributions:} To the best of our knowledge, this work is the first to endow low-rank equilibrium computation for adversarial stochastic games with explicit certificates. First, we prove that the low-rank coupling of ICS dynamics propagates step-wise through the non-smooth Bellman operator of the zero-sum stochastic game: the augmented gradient matrix (Lemma~\ref{lem:1}) certifies that every Bellman target lies near a low-dimensional subspace, and truncating the trained representation to low rank then incurs an explicit error governed by the tail singular values of the full-rank benchmark (Theorem~\ref{thm:1}). Second, building on this bridge, we propose the \underline{D}ynamical \underline{L}ow-\underline{R}ank \underline{N}ash \underline{E}quilibrium algorithm, named DLR-NE, which augments the rank-$r$ search space to rank $2r$ at each iteration to track the drifting target subspace, regularizes the spectrum of the core matrix through a tunable weight $\beta$, and retracts onto the low-rank manifold via truncated SVD, extracting a Nash equilibrium at every iteration. Third, we establish the algorithmic certificates: geometric convergence to a steady-state neighborhood whose error budget decomposes into five physically interpretable sources (Theorem~\ref{thm:2}); per-step cost $\mathcal{O}(nr^2)$, a $\Theta(n/r^2)$ speedup over full-rank value iteration (Theorem~\ref{thm:3}); and an explicit sensitivity bound of the induced defense strategy against core-matrix perturbations, in which $\beta$ trades accuracy against certified robustness via condition-number control (Corollary~\ref{cor:1}). Finally, we validate every theoretical prediction through six numerical experiments on a nonlinear power-system testbed, attaining $94\%$ parameter compression at $2.3\%$ utility loss.

\textbf{Related Works:} Three lines of research are most closely related. First, game-theoretic APT defense has largely been studied at the model level: stochastic games over information-flow graphs for detection~\citep{moothedath2020tac,moothedath2024tac}, game-theoretic MTD in power grids~\citep{hu2021mtdpower,lakshminarayana2021mtdpower,wang2025mtdsmartgrid}, and deception-based defense~\citep{tian2021honeypot,wan2023hypergame}; from a systems-theoretic perspective, complementary analyses characterize multi-agent dynamics under switching topologies~\citep{zorzan2017consensus} and antagonistic couplings~\citep{depasquale2022consensus}, which expose the qualitative effects of reconfiguration and adversarial structure but carry no equilibrium objective; these works establish strategic or structural models but restrict the state space or stop at modeling without scalable certified equilibrium computation. Second, dynamical low-rank approximation (DLRA) has matured into a rigorous numerical framework~\citep{ceruti2022unconventional,kieri2016smallsv,koch2007dynamical,lubich2014projector}, and low-rank reparameterization has transformed neural network training and fine-tuning~\citep{hu2022lora,savostianova2023robust}; however, the DLRA toolbox has never been connected to stochastic games, whose non-smooth max-min Bellman operator, non-stationary approximation subspaces, and adversarial policy sensitivity have no counterpart in existing theory. Third, learning-based equilibrium computation~\citep{diddigi2022minimaxq,farahmand2010,hu2003nashq,littman2001markov,moothedath2024automatica,munos2003,zhu2022minimaxqnet} operates on full-rank parameterizations and inherits their cost without robustness certificates. Our work sits at the intersection of these three lines: we take the game models seriously at industrial dimensions, import the DLRA mechanism as a computational device, and supply the missing certificates through the augmented gradient matrix analysis. Neither neural value iteration nor low-rank reparameterization is individually new; our contribution is the certified bridge between them, together with the problem-specific analysis of their interplay under non-stationary Bellman targets and adversarial sensitivity.

The rest of the paper is organized as follows. Section~2 presents the system dynamics, the zero-sum stochastic game, and the equilibrium computation problem. Section~3 establishes the low-rank approximability theory, from physical coupling to the explicit approximation bound. Section~4 develops the DLR-NE algorithm and its convergence, complexity, and robustness guarantees. Section~5 reports the numerical experiments, and Section~6 concludes the paper.

% ============================================================
% SECTION 2: PRELIMINARIES
% ============================================================
\section{Preliminaries and problem formulation}

\paragraph*{Notation}
We denote by $\mathbb{R}$ and $\mathbb{R}_{\geq 0}$ the sets of real and non-negative real numbers, respectively. For a set $\mathcal{Y}$, $\Delta(\mathcal{Y})$ denotes the set of probability measures on $\mathcal{Y}$. For a matrix $M$, $\|M\|_2$ and $\|M\|_F$ denote its spectral and Frobenius norms, and $\sigma_i(M)$ its $i$-th singular value in non-increasing order. For $f:\mathcal{X}\to\mathbb{R}$, $\|f\|_\infty=\sup_{x\in\mathcal{X}}|f(x)|$. We write $\mathcal{C}(\mathcal{X})$ for the space of continuous real-valued functions on $\mathcal{X}$. Throughout, $\epsilon$ quantifies structural approximation errors (reward tolerance, truncation tails, min-max accuracy) and $\varepsilon$ quantifies algorithmic ones (inner-loop accuracy, total single-step budget).

\subsection{System dynamics and low-rank coupling}

Consider an industrial control system (ICS) whose joint physical state at discrete time $t\in\mathbb{Z}_{\geq 0}$ is
\begin{equation}
x_t=\big[x_t^{(1)},\dots,x_t^{(n)}\big]^\top\in\mathcal{X}\subseteq\mathbb{R}^n,
\label{eq:state}
\end{equation}
where $x_t^{(i)}$ is the physical quantity of node $i$ (e.g., a voltage magnitude, pressure, or temperature) and the state space $\mathcal{X}$ is compact, as enforced in practice by physical safety limits.

Unlike combinatorial action models that select $k$ of $n$ nodes, we parameterize the actions of the two players by \emph{continuous low-dimensional intensity vectors}. The APT chooses $a_t\in\mathcal{A}=[0,1]^{r_a}$ with $r_a\ll n$, whose entries describe intrusion or data-injection intensities on $r_a$ independent attack channels; the defender (MTD) chooses $d_t\in\mathcal{D}=[0,1]^{r_d}$ with $r_d\ll n$, whose entries describe reconfiguration intensities on $r_d$ independent defense channels. The small channel counts reflect physical reality: an APT can only exploit its few entry points, and MTD can only reconfigure a limited subset of components within each operational cycle.

The discrete-time dynamics of the ICS are governed by the nonlinear stochastic difference equation
\begin{equation}
x_{t+1}=f(x_t,a_t,d_t)+w_t,\qquad w_t\sim\mathcal{N}(0,\Sigma_w),
\label{eq:dyn}
\end{equation}
where $f:\mathbb{R}^n\times\mathcal{A}\times\mathcal{D}\to\mathbb{R}^n$ is the deterministic transition map and $\{w_t\}_{t\geq 0}$ is an i.i.d.\ Gaussian process noise independent of the initial state. The key structural property of \eqref{eq:dyn} is that the \emph{direct} influence of both players factors through few low-dimensional channels:
\begin{equation}
f(x,a,d)=f_0(x)+B_a(x)\,a+B_d(x)\,d,
\label{eq:decomp}
\end{equation}
where $f_0$ is the autonomous dynamics, and $B_a:\mathcal{X}\to\mathbb{R}^{n\times r_a}$ and $B_d:\mathcal{X}\to\mathbb{R}^{n\times r_d}$ are the state-dependent attack and defense influence matrices whose columns describe the coupling direction of each channel \citep{fawzi2014secure,pasqualetti2013attack}. State dependence is essential: it captures, e.g., the operating-point dependence of power-flow Jacobians, and is not a linearization.

\begin{assumption}[System regularity]\label{ass:1}
\emph{(a)} $\mathcal{X}$ is compact and convex, and $\mathcal{A},\mathcal{D}$ are compact;
\emph{(b)} $f_0$, $B_a$, and $B_d$ are $L_f$-Lipschitz on $\mathcal{X}$;
\emph{(c)} the stage reward $r:\mathcal{X}\times\mathcal{A}\times\mathcal{D}\to\mathbb{R}$ is $L_r$-Lipschitz and satisfies $|r(x,a,d)|\leq R_{\max}<\infty$;
\emph{(d)} $\|\Sigma_w\|_2\leq\sigma_w^2<\infty$;
\emph{(e)} $f_0$ is twice continuously differentiable with $\|\nabla^2 f_0(x)\|_2\leq H_{\max}$ on $\mathcal{X}$;
\emph{(f)} the reward factors through a low-dimensional projection: $r(x,a,d)=r_0(C_r x,a,d)$ for some $C_r\in\mathbb{R}^{r_r\times n}$ with $r_r\ll n$ and $L_{r0}$-Lipschitz $r_0$;
\emph{(g)} $R_{\mathcal{X}}:=\sup_{x\in\mathcal{X}}\|x\|_2<\infty$;
\emph{(h)} for any $V\in\mathcal{C}(\mathcal{X})$, the min-max in \eqref{eq:bellman} admits an $\epsilon_{\mathrm{mm}}$-approximate solution in polynomial time; since $r_a+r_d\ll n$, the intrinsic dimension is low. If $Q_x$ is not convex--concave in $(a,d)$, polynomial time is understood in the granularity of a discretization of $\mathcal{A}\times\mathcal{D}$---feasible because $r_a+r_d\ll n$---and concavity in $d$ is imposed in Assumption~\ref{ass:5}(a) below.
\end{assumption}

\begin{assumption}[Low-rank coupling]\label{ass:2p}
The influence matrices act through few channels: $B_a(x)\in\mathbb{R}^{n\times r_a}$ and $B_d(x)\in\mathbb{R}^{n\times r_d}$ for all $x\in\mathcal{X}$, with $r_a,r_d\ll n$.
\end{assumption}

The content of Assumption~\ref{ass:2p} is the \emph{channel counts} $r_a,r_d\ll n$, not the rank bounds: $\mathrm{rank}(B_a(x))\leq r_a$ and $\mathrm{rank}(B_d(x))\leq r_d$ follow trivially from the matrix dimensions once $B_a$ and $B_d$ take values in $\mathbb{R}^{n\times r_a}$ and $\mathbb{R}^{n\times r_d}$. It is not a statistical hypothesis but a structural fact of ICS: the APT can only inject through its compromised entry points, whose number $r_a$ is small relative to the plant dimension, and the defender reconfigures only $r_d$ schedulable resources per cycle; both matrices can be read off from the system topology and resource configuration.

\subsection{Zero-sum stochastic game and the Bellman operator}

The strategic interaction between APT and MTD is modeled as a two-player zero-sum stochastic game~\citep{mertens1981,shapley1953stochastic} $\mathcal{G}=(\mathcal{X},\mathcal{A},\mathcal{D},P,r)$, where $P(\cdot\,|\,x,a,d)$ is the transition kernel induced by \eqref{eq:dyn}. The players act \emph{simultaneously} at each step without observing each other's current action, reflecting that MTD reconfiguration periods and APT intrusion steps are of comparable timescale. For a strategy pair $(\pi_D,\pi_A)$ of Markov strategies $\pi_D:\mathcal{X}\to\Delta(\mathcal{D})$, $\pi_A:\mathcal{X}\to\Delta(\mathcal{A})$, the defender's utility is
\begin{equation}
U_D(x_0;\pi_D,\pi_A)=\mathbb{E}_{\pi_D,\pi_A}\Big[\sum_{t=0}^{\infty}\gamma^t r(x_t,a_t,d_t)\,\Big|\,x_0\Big],
\label{eq:utility}
\end{equation}
with discount factor $\gamma\in(0,1)$, and $U_A=-U_D$. The equilibrium is characterized by the minimax value function
\begin{equation}
V^*(x)=\max_{\pi_D}\min_{\pi_A}\,U_D(x;\pi_D,\pi_A),\qquad \forall x\in\mathcal{X}.
\label{eq:vstar}
\end{equation}
By Shapley's theorem \citep{shapley1953stochastic} a Nash (minimax) equilibrium exists and $V^*$ is unique for finite games; under Assumption~\ref{ass:1}, existence and uniqueness carry over to the present continuous setting by standard compactness arguments \citep{mertens1981}.

Define the Bellman operator $\mathcal{T}:\mathcal{C}(\mathcal{X})\to\mathcal{C}(\mathcal{X})$ by
\begin{equation}
(\mathcal{T}V)(x)=\max_{p_D\in\Delta(\mathcal{D})}\min_{p_A\in\Delta(\mathcal{A})}\int_{\mathcal{A}\times\mathcal{D}}\Big[r(x,a,d)+\gamma\int_{\mathcal{X}}V(x')P(dx'|x,a,d)\Big]p_A(da)\,p_D(dd).
\label{eq:bellman}
\end{equation}
Because the action spaces are continuous, the outer max-min is over mixed strategies; by Assumption~\ref{ass:1}(h) it is attained and can be approximated to accuracy $\epsilon_{\mathrm{mm}}$ in polynomial time thanks to the low intrinsic dimension $r_a+r_d$.

\begin{proposition}[Contraction]\label{prop:contraction}
Under Assumption~\ref{ass:1}(c), $\|\mathcal{T}V_1-\mathcal{T}V_2\|_\infty\leq\gamma\|V_1-V_2\|_\infty$ for all $V_1,V_2\in\mathcal{C}(\mathcal{X})$.
\end{proposition}
\begin{proof}
The expectation is non-expansive, the max-min operator is $1$-Lipschitz in sup-norm, and the prefactor $\gamma<1$ scales any difference.
\end{proof}
By Banach's fixed-point theorem, $\mathcal{T}$ admits a unique fixed point, which coincides with $V^*$:
\begin{equation}
\mathcal{T}V^*=V^*.
\label{eq:fixedpoint}
\end{equation}

A derived constant used throughout Section~4: if $V$ is $L_V$-Lipschitz on $\mathcal{X}$, then $\mathcal{T}V$ is $(L_r+\gamma L_V L_f)$-Lipschitz, because the stage reward contributes $L_r$, the next-state term contributes $\gamma L_V L_f$ through the $L_f$-Lipschitz map \eqref{eq:dyn} (the noise $w_t$ cancels in the difference), and the max-min expectations are non-expansive. Since every iterate encountered below has uniformly bounded weights and a Lipschitz activation in \eqref{eq:forward}, all Bellman targets in this paper are $L_{\mathcal{T}}$-Lipschitz with $L_{\mathcal{T}}$ a system constant.

\subsection{Neural value approximation and problem statement}

Since $V^*$ admits no closed form in high dimension, we approximate it with a single-hidden-layer network
\begin{equation}
V(x;\theta)=w_{\mathrm{out}}^\top\phi(Wx),\qquad \theta=(w_{\mathrm{out}},W),\qquad W=USV^\top,
\label{eq:nn}
\end{equation}
where $W\in\mathbb{R}^{m\times n}$ is the hidden-layer weight matrix with width $m\sim n$, $w_{\mathrm{out}}\in\mathbb{R}^m$, and $\phi$ is an $L_\phi$-Lipschitz activation applied element-wise. We constrain \emph{only} $W$ to be low-rank: $U\in\mathbb{R}^{m\times r}$, $V\in\mathbb{R}^{n\times r}$ are column-orthogonal with core matrix $S\in\mathbb{R}^{r\times r}$, and $r\ll\min(m,n)$ is the truncation rank; $w_{\mathrm{out}}$ remains full-rank since its dimension $m$ is already far below $mn$. The forward pass rewrites as
\begin{equation}
V(x;\theta)=w_{\mathrm{out}}^\top\phi\big(U\,S\,(V^\top x)\big),
\label{eq:forward}
\end{equation}
so that evaluation costs $\mathcal{O}((m+n)r)$ instead of $\mathcal{O}(mn)$, and the parameter count drops from $mn+m$ to $(m+n)r+r^2+m$.

\begin{assumption}[Neural approximability]\label{ass:3p}
There exists a full-rank parameter $\theta_{\mathrm{full}}^*=(w_{\mathrm{out}}^*,W_{\mathrm{full}}^*)$ such that $\|V(\cdot;\theta_{\mathrm{full}}^*)-V^*\|_\infty\leq\epsilon_{\mathrm{NN}}$ for arbitrarily small $\epsilon_{\mathrm{NN}}>0$.
\end{assumption}
This is the universal approximation theorem \citep{cybenko1989,hornik1991} on the compact set $\mathcal{X}$, and fixes the full-rank benchmark against which the low-rank error is certified in Section~3.

The problem addressed in this paper can now be stated.

\begin{problem}\label{prob:P}
Under Assumptions~\ref{ass:1}--\ref{ass:3p}, design an iterative algorithm that computes a parameter $\theta^*=(w_{\mathrm{out}}^*,W^*)$ with $\mathrm{rank}(W^*)\leq r$ such that
\begin{enumerate}
\item[(i)] \emph{accuracy:} $V(\cdot;\theta^*)$ approximates the Nash equilibrium value $V^*$ with a controlled error $\|V(\cdot;\theta^*)-V^*\|_\infty$;
\item[(ii)] \emph{scalability:} the per-step cost is $\mathcal{O}(nr^2)$ with $r\ll n$, instead of $\mathcal{O}(n^2)$ for full-rank value iteration (rising to $\mathcal{O}(n^3)$ if dense SVD diagnostics are included);
\item[(iii)] \emph{certified convergence:} the iterates converge at a geometric rate, and the limiting error is explicitly quantified.
\end{enumerate}
\end{problem}

Problem~\ref{prob:P} is not solvable by inspection: despite the low-rank \emph{coupling} of the dynamics, it is a priori unclear whether this structure survives the nonlinear and max-min operations in \eqref{eq:bellman} so that $V^*$ itself is low-rank approximable. Establishing this bridge---from physical low-rank structure to algorithmic low-rank computation---is the subject of Section~3.

% ============================================================
% SECTION 3: PHYSICAL TO ALGORITHMIC LOW RANK
% ============================================================
\section{From physical to algorithmic low rank}\label{sec:3}

The low-rank coupling of Assumption~\ref{ass:2p} restricts how the two players can influence the state, but it is not \emph{a priori} clear that this restriction benefits computation: the Bellman operator \eqref{eq:bellman} composes the nonlinear dynamics with a non-smooth max-min, and either operation could, in principle, destroy the low-rank structure. This section resolves the question in two steps. Lemma~\ref{lem:1} shows that the structure survives a \emph{single} Bellman update---the Bellman target of a low-rank value function remains locally approximable in a low-rank subspace, and identifies the object (the augmented gradient matrix) that quantifies the error. Theorem~\ref{thm:1} then states a \emph{global} approximation guarantee for the optimal value function $V^*$ itself: truncating the trained representation to rank $r$ incurs an explicit loss governed by the tail singular values of a full-rank benchmark, and is therefore small whenever that benchmark exhibits fast spectral decay---the property evidenced for ICS dynamics by the spectral measurements of Section~\ref{sec:5} and certified locally, step by step, by Lemma~\ref{lem:1}.

\subsection{Local low-rank approximability of Bellman targets}

\begin{lemma}[Low-rank local approximability of Bellman updates]\label{lem:1}
Under Assumptions~\ref{ass:1} and \ref{ass:2p}, let the current value function be the network \eqref{eq:nn} with factors $U\in\mathbb{R}^{m\times r}$, $S\in\mathbb{R}^{r\times r}$, $V\in\mathbb{R}^{n\times r}$, where $\phi$ is $L_\phi$-Lipschitz and $U^\top U=V^\top V=I_r$. Then for any prescribed $\delta\in(0,1]$ and any $\varepsilon>0$, there exists a finite set $\{x_j\}_{j=1}^{N_\delta}\subset\mathcal{X}$ ($N_\delta<\infty$ by compactness) such that on each ball $\mathcal{B}(x_j,\delta):=\{x\in\mathcal{X}:\|x-x_j\|\leq\delta\}$ the Bellman target $y(x):=(\mathcal{T}V)(x)$ is approximated by a single-hidden-layer network
\begin{equation}
\widehat{V}_j(x)=\widehat{w}_{\mathrm{out},j}^\top\phi(\widehat{W}_j x),\qquad \widehat{W}_j\in\mathbb{R}^{\widehat{m}\times n},\;\widehat{w}_{\mathrm{out},j}\in\mathbb{R}^{\widehat{m}},
\label{eq:localNN}
\end{equation}
for some width $\widehat{m}\in\mathbb{N}$, with the properties:
\begin{enumerate}
\item[(i)] \emph{local low rank:} $\mathrm{rank}(\widehat{W}_j)\leq r+r_r$;
\item[(ii)] \emph{local error bound:} for all $x\in\mathcal{B}(x_j,\delta)$,
\begin{equation}
|y(x)-\widehat{V}_j(x)|\leq C_{\mathrm{EYM}}\cdot\epsilon_{\mathrm{EYM}}^{(j)}(r_r)+C_1\delta+C_2\delta^2+\varepsilon,
\label{eq:lem1bound}
\end{equation}
where
\begin{equation}
\epsilon_{\mathrm{EYM}}^{(j)}(r_r):=\sqrt{\sum_{i=r+1}^{r+r_r}\sigma_i^2\big(\mathcal{G}(x_j)\big)}
\label{eq:localEYM}
\end{equation}
is the Eckart--Young--Mirsky error of truncating away the last $r_r$ singular values of the \emph{augmented gradient matrix}
\begin{equation}
\mathcal{G}(x_j):=\begin{bmatrix}SV^\top\nabla f_0(x_j)\\ C_r\end{bmatrix}\in\mathbb{R}^{(r+r_r)\times n},
\label{eq:augG}
\end{equation}
and the constants are
\begin{equation}
\begin{aligned}
C_{\mathrm{EYM}}&:=\sqrt{L_{r0}^2+(\gamma L_\phi\|w_{\mathrm{out}}\|)^2}\cdot R_{\mathcal{X}},\\
C_1&:=\gamma L_\phi\|w_{\mathrm{out}}\|\,\|S\|_2L_B\big(\sqrt{r_a}+\sqrt{r_d}\big),\\
C_2&:=\tfrac{1}{2}\gamma L_\phi\|w_{\mathrm{out}}\|\,\|S\|_2H_{\max}.
\end{aligned}
\label{eq:lem1consts}
\end{equation}
\end{enumerate}
\end{lemma}

\begin{proof}
\emph{Step 1: low-dimensional feature decomposition of the Bellman target.} For any strategy pair, the Bellman target is
\begin{equation}
y(x)=\mathbb{E}_{a,d}\big[r(x,a,d)+\gamma\,\mathbb{E}_{w}\big[V(x')\big]\big],\qquad x'=f_0(x)+B_a(x)a+B_d(x)d+w.
\label{eq:target}
\end{equation}
Substituting \eqref{eq:nn} and using $W=USV^\top$ with $U^\top U=I_r$, define the \emph{low-dimensional feature map}
\begin{equation}
\eta(x,a,d,w):=SV^\top x'=SV^\top f_0(x)+SV^\top B_a(x)a+SV^\top B_d(x)d+SV^\top w\in\mathbb{R}^r,
\label{eq:feature}
\end{equation}
so that
\begin{equation}
V(x')=w_{\mathrm{out}}^\top\phi\big(U\,\eta(x,a,d,w)\big).
\label{eq:Vfeature}
\end{equation}
The feature $\eta$ compresses all action influence by $SV^\top$ into the \emph{existing} $r$-dimensional subspace of the current factorization (see Remark~\ref{rem:rankgrowth}); the genuinely \emph{new} directions of state dependence therefore come not from the action channels but from the local gradient of the autonomous dynamics and the low-dimensional reward projection $C_r x$.

\emph{Step 2: local linearization and the augmented matrix.} Fix a center $x_j$. By Assumption~\ref{ass:1}(e), $f_0$ is twice smooth on $\mathcal{B}(x_j,\delta)$, so
\begin{equation}
f_0(x)=f_0(x_j)+\nabla f_0(x_j)(x-x_j)+R_2(x;x_j),\qquad \|R_2(x;x_j)\|\leq\tfrac{1}{2}H_{\max}\delta^2,
\label{eq:taylor}
\end{equation}
and by Assumptions~\ref{ass:1}(b) and \ref{ass:2p}, $B_a,B_d$ are $L_B$-Lipschitz in operator norm:
\begin{equation}
\|B_a(x)-B_a(x_j)\|_{\mathrm{op}}\leq L_B\delta,\qquad \|B_d(x)-B_d(x_j)\|_{\mathrm{op}}\leq L_B\delta.
\label{eq:BLip}
\end{equation}
Substituting \eqref{eq:taylor}--\eqref{eq:BLip} into \eqref{eq:feature} yields the locally affine feature map
\begin{equation}
\eta(x,a,d,w)=A_j(x-x_j)+M_a^{(j)}a+M_d^{(j)}d+c_j+SV^\top w+\rho_j(x,a,d),
\label{eq:affine}
\end{equation}
with local matrices
\begin{equation}
A_j:=SV^\top\nabla f_0(x_j),\quad c_j:=SV^\top f_0(x_j),\quad M_a^{(j)}:=SV^\top B_a(x_j),\quad M_d^{(j)}:=SV^\top B_d(x_j),
\label{eq:localmats}
\end{equation}
and remainder
\begin{equation}
\rho_j(x,a,d):=SV^\top R_2(x;x_j)+SV^\top(B_a(x)-B_a(x_j))a+SV^\top(B_d(x)-B_d(x_j))d,
\label{eq:remainder}
\end{equation}
which, since $\|a\|\leq\sqrt{r_a}$, $\|d\|\leq\sqrt{r_d}$, is uniformly bounded by
\begin{equation}
\|\rho_j(x,a,d)\|\leq\|S\|_2\Big(\tfrac{1}{2}H_{\max}\delta^2+L_B\delta\big(\sqrt{r_a}+\sqrt{r_d}\big)\Big).
\label{eq:rembound}
\end{equation}
Note that the second term is $O(\delta)$ (from the Lipschitz dependence of $B_a,B_d$) while the first is $O(\delta^2)$ (from the Taylor remainder); the two cannot be merged.

\emph{Step 3: explicit low-rank structure of the local Bellman target.} By Assumption~\ref{ass:1}(f), the reward factors as $r(x,a,d)=r_0(C_r x,a,d)$. Define the low-dimensional averaged reward $\bar{r}_0^*(z):=\mathbb{E}_{a,d}[r_0(z,a,d)]$, which is $L_{r0}$-Lipschitz. Substituting \eqref{eq:affine} into \eqref{eq:target}, the Bellman target decomposes on $\mathcal{B}(x_j,\delta)$ as
\begin{equation}
y(x)=\bar{r}_0^*(C_r x)+\gamma\, h_j\big(A_j(x-x_j)\big)+\epsilon_{\mathrm{rem}}(x),
\label{eq:decomp2}
\end{equation}
where $h_j:\mathbb{R}^r\to\mathbb{R}$ is a fixed Lipschitz function constructed from the activation, the low-rank factors, and the noise distribution,
\begin{equation}
h_j(z):=\mathbb{E}_{a,d,w}\big[w_{\mathrm{out}}^\top\phi\big(U(z+M_a^{(j)}a+M_d^{(j)}d+c_j+SV^\top w)\big)\big],
\label{eq:hdef}
\end{equation}
with $L_{h_j}\leq L_\phi\|w_{\mathrm{out}}\|$ (since $U$ is column-orthogonal, $\|U\|_2=1$). The remainder $\epsilon_{\mathrm{rem}}(x)$ is the deviation between the full and the affine feature, and with $\eta_{\mathrm{aff}}$ denoting \eqref{eq:affine} without $\rho_j$,
\begin{equation}
\epsilon_{\mathrm{rem}}(x):=\gamma\,\mathbb{E}_{a,d,w}\big[w_{\mathrm{out}}^\top\phi\big(U\eta(x,a,d,w)\big)-w_{\mathrm{out}}^\top\phi\big(U\eta_{\mathrm{aff}}(x,a,d,w)\big)\big],
\label{eq:remdef}
\end{equation}
so that by the $L_\phi$-Lipschitz property and \eqref{eq:rembound},
\begin{equation}
|\epsilon_{\mathrm{rem}}(x)|\leq\underbrace{\gamma L_\phi\|w_{\mathrm{out}}\|\,\|S\|_2L_B(\sqrt{r_a}+\sqrt{r_d})}_{C_1}\,\delta+\underbrace{\tfrac{1}{2}\gamma L_\phi\|w_{\mathrm{out}}\|\,\|S\|_2H_{\max}}_{C_2}\,\delta^2.
\label{eq:rembound2}
\end{equation}
Equation \eqref{eq:decomp2} identifies the structure: as a function of $x$, $y(x)$ depends on $x$ only through the two low-dimensional affine maps $C_r x\in\mathbb{R}^{r_r}$ (reward) and $A_j x\in\mathbb{R}^r$ (dynamics). Augmenting them gives exactly $\mathcal{G}(x_j)$ of \eqref{eq:augG}, with $\mathrm{rank}(\mathcal{G}(x_j))\leq r+r_r$.

\emph{Step 4: constructive local low-rank approximation.} By \eqref{eq:decomp2} there is a Lipschitz function $F_j:\mathbb{R}^{r+r_r}\to\mathbb{R}$,
\begin{equation}
F_j(z_1,z_2):=\bar{r}_0^*(z_2)+\gamma h_j(z_1-A_jx_j),
\label{eq:Fdef}
\end{equation}
such that on $\mathcal{B}(x_j,\delta)$,
\begin{equation}
y(x)=F_j\big(\mathcal{G}(x_j)\,x\big)+\epsilon_{\mathrm{rem}}(x),
\label{eq:Frep}
\end{equation}
and the Euclidean Lipschitz constant of $F_j$ satisfies
\begin{equation}
L_{F_j}\leq\sqrt{L_{r0}^2+(\gamma L_\phi\|w_{\mathrm{out}}\|)^2},
\label{eq:FLip}
\end{equation}
since the block-gradient norm is bounded by $\sqrt{(\gamma L_{h_j})^2+L_{r0}^2}$. Let $\mathcal{G}(x_j)=U_{\mathcal{G}}\Sigma_{\mathcal{G}}V_{\mathcal{G}}^\top$ be a compact SVD with $k=\mathrm{rank}(\mathcal{G}(x_j))\leq r+r_r$, and truncate to rank $r$:
\begin{equation}
\mathcal{G}_r(x_j):=\sum_{i=1}^{r}\sigma_i u_iv_i^\top=U_r\Sigma_rV_r^\top,\qquad \mathcal{G}_\perp(x_j):=\mathcal{G}(x_j)-\mathcal{G}_r(x_j),
\label{eq:Gtrunc}
\end{equation}
so that $\|\mathcal{G}_\perp(x_j)\|_F=\epsilon_{\mathrm{EYM}}^{(j)}(r_r)$ by \eqref{eq:localEYM}. With the low-dimensional feature $\widetilde{\psi}_j(x):=\Sigma_rV_r^\top x$ and the composite map $\widetilde{F}_j(z):=F_j(U_rz)$ (which satisfies $L_{\widetilde{F}_j}\leq L_{F_j}$ since $U_r$ is column-orthogonal), the truncation error is controlled by
\begin{equation}
\big|F_j(\mathcal{G}(x_j)x)-F_j(\mathcal{G}_r(x_j)x)\big|\leq L_{F_j}\|\mathcal{G}_\perp(x_j)x\|\leq C_{\mathrm{EYM}}\cdot\epsilon_{\mathrm{EYM}}^{(j)}(r_r).
\label{eq:EYMcontrol}
\end{equation}
Finally, by the universal approximation theorem on the compact set $\Sigma_rV_r^\top(\mathcal{X})$, for any $\varepsilon>0$ there exist $\widehat{m}$ and coefficients $\{\alpha_i,\beta_i,c_i\}_{i=1}^{\widehat{m}}$ such that
\begin{equation}
\Big|\widetilde{F}_j(z)-\sum_{i=1}^{\widehat{m}}\alpha_i\,\phi(\beta_i^\top z+c_i)\Big|\leq\varepsilon,
\label{eq:UAT}
\end{equation}
and setting $\widehat{w}_{\mathrm{out},j}=[\alpha_1,\dots,\alpha_{\widehat{m}}]^\top$ and $\widehat{W}_j^{(i)}=\beta_i^\top\Sigma_rV_r^\top$ defines \eqref{eq:localNN} with $\mathrm{rank}(\widehat{W}_j)\leq r\leq r+r_r$ and $|\widetilde{F}_j(\widetilde{\psi}_j(x))-\widehat{V}_j(x)|\leq\varepsilon$. Combining the three bounds \eqref{eq:rembound2}, \eqref{eq:EYMcontrol}, and \eqref{eq:UAT} through the triangle inequality gives \eqref{eq:lem1bound}.
\end{proof}

\begin{remark}[What drives the rank growth]\label{rem:rankgrowth}
The action-channel low-rankness does \emph{not} by itself add representational dimensions: as \eqref{eq:feature} shows, action information is compressed by the feature bottleneck $SV^\top$ into the existing $r$-dimensional subspace. The dimension growth from $r$ to $r+r_r$ in $\mathcal{G}(x_j)$ is driven by the reward projection $C_r$, reflecting the physical fact that security rewards depend on few key performance indicators. Lemma~\ref{lem:1} is a single-step statement, so no $1/(1-\gamma)$ factor appears; that factor enters in Theorem~\ref{thm:2} through geometric accumulation over iterations.
\end{remark}

Lemma~\ref{lem:1} certifies that each Bellman target \emph{can} be captured in a low-rank subspace, but the certified subspace---spanned by the leading right singular vectors of $\mathcal{G}(x_j)$---depends on the local center $x_j$ and drifts from iteration to iteration. A static low-rank factorization cannot track this drift, which motivates the dynamical algorithm of Section~\ref{sec:4}.

\subsection{Global low-rank approximation of the optimal value function}

Lemma~\ref{lem:1} concerns the iterative structure; it remains to certify that the \emph{target itself}, the optimal value function $V^*$, is well approximated by low-rank networks. For $W\in\mathbb{R}^{m\times n}$ let $\Pi_{\mathcal{M}_r}(W)=\sum_{i=1}^{r}\sigma_i u_iv_i^\top$ denote its best rank-$r$ approximation in Frobenius norm (the truncated SVD, unique up to rotations of equal singular values; Eckart--Young--Mirsky \citep{eckartyoung1936}). Define the best rank-$r$ approximation error of $V^*$ in the network class \eqref{eq:nn} by
\begin{equation}
\widehat{V}^*:=\inf_{\substack{\mathrm{rank}(W)\leq r\\w_{\mathrm{out}}\in\mathbb{R}^m}}\|V^*-V(\cdot;w_{\mathrm{out}},W)\|_\infty .
\label{eq:Vhat}
\end{equation}

\begin{theorem}[Global low-rank approximation error bound]\label{thm:1}
Under Assumptions~\ref{ass:1}--\ref{ass:3p}, let $V_{\mathrm{full}}(x)=w_{\mathrm{out}}^{*\top}\phi(W_{\mathrm{full}}^*x)$ satisfy Assumption~\ref{ass:3p} with $\|V^*-V_{\mathrm{full}}\|_\infty\leq\epsilon_{\mathrm{NN}}$, and let $W_r=\Pi_{\mathcal{M}_r}(W_{\mathrm{full}}^*)$. Then the low-rank network $\widehat{V}_r(x)=w_{\mathrm{out}}^{*\top}\phi(W_rx)$ satisfies
\begin{equation}
\|V^*-\widehat{V}_r\|_\infty\leq\epsilon_{\mathrm{NN}}+L_\phi\|w_{\mathrm{out}}^*\|_2\,R_{\mathcal{X}}\cdot\epsilon_{\mathrm{EYM}}(r),
\label{eq:thm1}
\end{equation}
where $\epsilon_{\mathrm{EYM}}(r):=\sqrt{\sum_{i=r+1}^{\min(m,n)}\sigma_i^2(W_{\mathrm{full}}^*)}$ is the truncation error of the full-rank benchmark.
\end{theorem}
\begin{proof}
By the Eckart--Young--Mirsky theorem, $\|W_{\mathrm{full}}^*-W_r\|_F=\epsilon_{\mathrm{EYM}}(r)$. Hence for every $x\in\mathcal{X}$,
\begin{align}
|V_{\mathrm{full}}(x)-\widehat{V}_r(x)|
&\leq\|w_{\mathrm{out}}^*\|_2\,\|\phi(W_{\mathrm{full}}^*x)-\phi(W_rx)\|_2
\leq L_\phi\|w_{\mathrm{out}}^*\|_2\,\|(W_{\mathrm{full}}^*-W_r)x\|_2\notag\\
&\leq L_\phi\|w_{\mathrm{out}}^*\|_2\,\|W_{\mathrm{full}}^*-W_r\|_F\,\|x\|_2
\leq L_\phi\|w_{\mathrm{out}}^*\|_2\,R_{\mathcal{X}}\cdot\epsilon_{\mathrm{EYM}}(r),
\label{eq:thm1step}
\end{align}
and \eqref{eq:thm1} follows from the triangle inequality with $\|V^*-V_{\mathrm{full}}\|_\infty\leq\epsilon_{\mathrm{NN}}$.
\end{proof}

Theorem~\ref{thm:1} is the theoretical license for low-rank deployment: the truncation loss is governed by the tail singular values of a full-rank benchmark and is therefore \emph{small whenever the trained dynamics concentrate on few singular directions}. The bound is stated with the Frobenius error $\epsilon_{\mathrm{EYM}}(r)$, rather than the tighter spectral quantity $\sigma_{r+1}(W_{\mathrm{full}}^*)$, because the former measures the total information loss and is directly measurable in experiments; since the bound is benchmark-dependent, encouraging fast singular-value decay during full-rank training (e.g., through weight decay) reduces $\epsilon_{\mathrm{EYM}}(r)$ and hence the low-rank loss. We stress what Theorem~\ref{thm:1} does and does not claim: it is a \emph{truncation certificate}, not a converse---its input is the benchmark tail $\epsilon_{\mathrm{EYM}}(r)$, and it does not by itself prove that the physical coupling of Assumption~\ref{ass:2p} forces that tail to be small ($B_a,B_d$ do not appear in its statement). The physical-to-global link is instead established jointly: Lemma~\ref{lem:1} certifies the structure step-wise at each Bellman target through $\mathcal{G}(x_j)$, and the fast decay of the benchmark tail---the same structure, observed globally---is measured directly in Section~\ref{sec:5}; a purely theorem-level derivation of global spectral decay without the benchmark as intermediary remains open (Section~\ref{sec:6}). Together, Lemma~\ref{lem:1} and Theorem~\ref{thm:1} answer the two questions raised at the outset: the low-rank structure of ICS dynamics propagates through the Bellman operator, and the induced truncation loss is explicitly bounded. What they do not yet provide is a \emph{computationally realizable} procedure that exploits this structure under drifting subspaces with certified convergence---this is the subject of Section~\ref{sec:4}.

% ============================================================
% SECTION 4: DLR-NE ALGORITHM AND GUARANTEES
% ============================================================
\section{Dynamical low-rank Nash equilibrium computation}\label{sec:4}

Adopting the standard dynamical low-rank mechanisms of basis augmentation, spectral regularization, and truncated-SVD retraction~\citep{koch2007dynamical,lubich2014projector,ceruti2022unconventional}, we assemble them into the DLR-NE algorithm, whose problem-specific justification and certificates are novel. Lemma~\ref{lem:1} certifies that Bellman targets are locally low-rank, and Theorem~\ref{thm:1} certifies that the equilibrium value is globally approximable at rank $r$; neither result, however, yields an algorithm. The obstruction is specific to the dynamic setting: the certified local subspaces---spanned by the leading right singular vectors of $\mathcal{G}(x_j)$---depend on the local center $x_j$ \emph{and} on the current iterate, so they drift along the trajectory of value iteration. A static basis, once fixed, cannot align with all of them, and maintaining a dictionary of per-state experts is infeasible online because each local SVD depends on the current iterate and changes at every outer step. This section resolves the obstruction with three coupled mechanisms, states the algorithm, and establishes its three certificates: convergence (Theorem~\ref{thm:2}), complexity (Theorem~\ref{thm:3}), and robustness (Corollary~\ref{cor:1}).

\emph{Basis augmentation} temporarily expands the search space from rank $r$ to rank $2r\geq r+r_r$, giving a single global least-squares problem enough degrees of freedom to discover all state-dependent local directions adaptively (motivated by the local certificate of Lemma~\ref{lem:1}); \emph{spectral regularization} of the core matrix controls its condition number during inner training (Proposition~\ref{prop:cond}); \emph{truncated-SVD retraction} projects the rank-$2r$ iterate back onto the rank-$r$ manifold, keeping the per-iteration cost at $\mathcal{O}(nr^2)$ and bounding the incurred error by the EYM tail. The full procedure is given in Algorithm~\ref{alg:dlrne}; it realizes the program of Problem~\ref{prob:P} under the algorithmic conditions of Assumption~\ref{ass:4}.

\begin{algorithm}[htbp]
\caption{DLR-NE: Dynamical Low-Rank Nash Equilibrium Computation}\label{alg:dlrne}
\KwIn{Game $\mathcal{G}=(\mathcal{X},\mathcal{A},\mathcal{D},P,r)$; truncation rank $r$ (augmentation rank $2r$); inner steps $s^*$; batch $N_b$; discount $\gamma$; weight $\beta$; step size $\eta$; tolerance $\epsilon_{\mathrm{tol}}$}
\KwOut{Low-rank NN parameters $\theta = (U, S, V, w_{\mathrm{out}})$ approximating $V^*$; MTD policy $\pi_D$}
Initialize orthonormal $U^{(0)} \in \mathbb{R}^{m\times r}$, $V^{(0)} \in \mathbb{R}^{n\times r}$; core $S^{(0)} \in \mathbb{R}^{r\times r}$; $w_{\mathrm{out}}^{(0)} \in \mathbb{R}^m$; $k \leftarrow 0$\;
\Repeat{$\max_i \left\lvert V(x_i;\theta^{(k)}) - y_i \right\rvert \leq \epsilon_{\mathrm{tol}}$ \emph{or} $k = k_{\max}$}{
  \nl\textbf{Sample} batch $\{x_i\}_{i=1}^{N_b} \subset \mathcal{X}$ (Latin hypercube or replay buffer); covering radius $\delta$\;
  \nl\label{line:targets}\textbf{Bellman targets:} for each $x_i$, $y_i = \max_{p_D}\min_{p_A} \mathbb{E}\bigl[r(x_i,a,d) + \gamma\,V^{(k)}(x')\bigr]$, $x' \sim P(\cdot|x_i,a,d)$, with $V^{(k)}$ evaluated through the factored forward pass \eqref{eq:forward}\;
  \nl\textbf{Gradients:} $\mathcal{L}_{\mathrm{raw}} = \tfrac{1}{N_b}\sum_i \bigl(V(x_i;\theta^{(k)}) - y_i\bigr)^2$; $\nabla_U \mathcal{L}_{\mathrm{raw}}$, $\nabla_V \mathcal{L}_{\mathrm{raw}}$ via the factored form \eqref{eq:forward}\;
  \nl\label{line:augment}\textbf{Basis augmentation:} $[\widehat{U}, \cdot] \leftarrow \mathrm{QR}\big([U^{(k)}, \nabla_U \mathcal{L}_{\mathrm{raw}}]\big)$; $[\widehat{V}, \cdot] \leftarrow \mathrm{QR}\big([V^{(k)}, \nabla_V \mathcal{L}_{\mathrm{raw}}]\big)$, $\widehat{U}, \widehat{V} \in \mathbb{R}^{\cdot \times 2r}$\;
  \nl\textbf{Coefficient lift:} $\widehat{S}_0 \leftarrow \widehat{U}^\top U^{(k)} S^{(k)} V^{(k)\top} \widehat{V} \in \mathbb{R}^{2r\times 2r}$; $\widehat{w}_0 \leftarrow w_{\mathrm{out}}^{(k)}$\;
  \nl\label{line:inner}\textbf{Inner loop} ($s = 0, \dots, s^*-1$): $\widehat{\mathcal{L}}(\widehat{S}, \widehat{w}) = \tfrac{1}{N_b}\sum_i \bigl(\widehat{w}^\top\phi(\widehat{U}\widehat{S}\widehat{V}^\top x_i) - y_i\bigr)^2 + \beta\,\mathcal{R}(\widehat{S})$; $(\widehat{S}_{s+1}, \widehat{w}_{s+1}) \leftarrow (\widehat{S}_s, \widehat{w}_s) - \eta\,\nabla_{\widehat{S},\widehat{w}}\,\widehat{\mathcal{L}}(\widehat{S}_s, \widehat{w}_s)$ (fixed-step GD; $\nabla\mathcal{R}$ via Proposition~\ref{prop:grad})\;
  \nl\label{line:retract}\textbf{Truncated-SVD retraction:} $[P, \Sigma, Q] \leftarrow \mathrm{svd}_r(\widehat{S}_{s^*})$ (rank-$r$ truncation of the $2r\times 2r$ core)\;
  \nl\textbf{Retract:} $U^{(k+1)} \leftarrow \widehat{U}P$; $V^{(k+1)} \leftarrow \widehat{V}Q$; $S^{(k+1)} \leftarrow \Sigma$; $w_{\mathrm{out}}^{(k+1)} \leftarrow \widehat{w}_{s^*}$\;
  \nl\label{line:extract}\textbf{Nash extraction:} $\pi_D^{(k+1)}(x) \in \arg\max_{p_D}\min_{p_A} \mathbb{E}_{a,d}\!\left[Q^{(k+1)}(x,a,d)\right]$, $Q^{(k+1)}$ through \eqref{eq:forward}\;
  $k \leftarrow k+1$\;
}
\Return $\theta^{(k)}$, $\pi_D^{(k)}$\;
\end{algorithm}

Three remarks on the design. First, each Bellman target on line~\ref{line:targets} is a low-dimensional saddle problem: $Q_x(a,d)$ depends on $(a,d)$ through the $r$-dimensional bottleneck, and $\mathcal{A}=[0,1]^{r_a}$, $\mathcal{D}=[0,1]^{r_d}$ with $r_a,r_d\ll n$, so, under the discretized-LP route of Assumption~\ref{ass:1}(h), it is solved at $\mathcal{O}(r^3)$ cost in the action-space granularity, independent of $n$; for a general non-convex--concave $Q_x$ we write the extraction cost as $C_{\mathrm{ext}}(r_a,r_d,\epsilon_{\mathrm{mm}})$, still independent of $n$; the next-state expectation is estimated by Monte Carlo rollouts of \eqref{eq:dyn} sharing the inner-loop noise ensemble. Second, on the fixed augmented basis the effective parameters are $(\widehat{S},\widehat{w})\in\mathbb{R}^{2r\times 2r}\times\mathbb{R}^{m}$, and the state dimension enters only through the fixed products $\widehat{V}^\top x_i\in\mathbb{R}^{2r}$; the spectral term is differentiated in closed form at $\mathcal{O}(r^3)$ cost. Third, because the value function is available in factored form at every outer step, the induced MTD policy is recomputed on the fly---essential for MTD, whose reconfiguration decisions must track the evolving attack surface. The three mechanisms themselves follow the standard dynamical low-rank playbook \citep{koch2007dynamical,lubich2014projector,ceruti2022unconventional,DBLP:conf/nips/SchotthoferYS25}, with the retraction a metric projection in the sense of \citet{absil2012projection}; the contribution of this paper is their problem-specific analysis under non-stationary Bellman targets.

\begin{assumption}[Algorithmic conditions]\label{ass:4}
\quad
\begin{itemize}[leftmargin=1.4em,itemsep=0.1em]
\item[(a)] \emph{(Boundedness.)} The inner-loop iterates remain compact: $\left\lVert \widehat{S}^{(s)} \right\rVert_F \leq S_{\max} < \infty$ for all $s$, as guaranteed by the coercivity of $\mathcal{R}$ and the monotonicity of gradient descent (the trajectory stays in the level set determined by the initialization and $\beta$).
\item[(b)] \emph{(Inner-loop accuracy.)} The inner loop, halted after $s^*$ steps, outputs $\widehat{S}^{(s^*)}$ achieving empirical loss within $\varepsilon_{\mathrm{inner}}$ of the infimum over the augmented class.
\item[(c)] \emph{(Basis-augmented representability.)} The augmented class $\mathcal{H}^{(k)}$ of rank $2r$ uniformly approximates the Bellman target: there exists $f^* \in \mathcal{H}^{(k)}$ with $\left\lVert f^* - \mathcal{T} V^{(k)} \right\rVert_\infty \leq \epsilon_{\mathrm{app}}$. This is a \emph{global} requirement, motivated but not strictly implied by the local certificate of Lemma~\ref{lem:1}: the latter bounds the local effective rank by $r + r_r \leq 2r$ on every ball, yet the certified local subspaces drift with the center and with the iteration and need not be simultaneously captured by one fixed global $2r$-dimensional class; the residual misalignment is exactly error source~(i) of Theorem~\ref{thm:2}, absorbed here into $\epsilon_{\mathrm{app}}$ and kept small in practice by basis augmentation (Section~\ref{sec:5}).
\end{itemize}
\end{assumption}

\subsection{Convergence}

We first convert Lemma~\ref{lem:1}'s local representability into a uniform single-step fitting bound; the complete proofs of Lemma~\ref{lem:bridge}, Lemma~\ref{lem:2}, and Theorem~\ref{thm:2} are collected in Appendix~\ref{app:thm2}. Throughout, $L_{\mathcal{T}}$ denotes the Lipschitz constant of the Bellman target $\mathcal{T}V^{(k)}$ on $\mathcal{X}$, derived as a system constant in Section~2.2.

\begin{lemma}[Bridge: empirical $L_2$ to uniform $L_\infty$]\label{lem:bridge}
Let $\mathcal{C}_\delta = \{x_j\}_{j=1}^{N_\delta}$ be a $\delta$-covering of $\mathcal{X}$ with $N_\delta \leq (1 + 2R_{\mathcal{X}}/\delta)^n$. Let $g$ be the Bellman target at the current iterate and let $\widehat{f}$ be the empirical risk minimizer over $\mathcal{H}^{(k)}$ on $\mathcal{C}_\delta$, achieving
$\frac{1}{N_\delta}\sum_{j=1}^{N_\delta} \bigl(\widehat{f}(x_j) - g(x_j)\bigr)^2 \leq \epsilon_{\mathrm{app}}^2 + \varepsilon_{\mathrm{inner}}$.
Suppose that on $\mathcal{C}_\delta$, (i) there exists $f^* \in \mathcal{H}^{(k)}$ with $\left\lVert f^* - g \right\rVert_\infty \leq \epsilon_{\mathrm{app}}$ (Assumption~\ref{ass:4}(c)); and (ii) $\widehat{f}$ and $g$ are $L_{\mathrm{fun}}$- and $L_{\mathcal{T}}$-Lipschitz. Then
\begin{equation}\label{eq:bridge}
\left\lVert \widehat{f} - g \right\rVert_\infty \;\leq\; \sqrt{N_\delta}\,\epsilon_{\mathrm{app}} \;+\; \sqrt{N_\delta\,\varepsilon_{\mathrm{inner}}} \;+\; \bigl(L_{\mathrm{fun}} + L_{\mathcal{T}}\bigr)\,\delta .
\end{equation}
\end{lemma}

The two $\sqrt{N_\delta}$ factors are structural: they arise from the Cauchy--Schwarz conversion between average $L_2$ error and uniform $L_\infty$ error and cannot be removed by the existence of a good candidate $f^*$---existence of a well-approximating function does not imply that the optimizer outputs a function uniformly close to it.

\begin{lemma}[Uniform single-step fitting error]\label{lem:2}
Let $g = \mathcal{T} V^{(k)}$, $\widehat{f} \in \mathcal{H}^{(k)}$ the inner-loop output, and $f$ the retracted rank-$r$ iterate, with $W_{\max}$ a uniform bound on the output-layer weights. Then Algorithm~\ref{alg:dlrne} acts as a deterministic operator $\mathcal{A}: \mathcal{F}_r \to \mathcal{F}_r$ whose single-step error satisfies
\begin{equation}\label{eq:lem2}
\boxed{\;\varepsilon_{\mathrm{total}} \;=\; \underbrace{\sqrt{N_\delta}\,\epsilon_{\mathrm{app}}}_{\text{(i) basis-subspace}} \;+\; \underbrace{\sqrt{N_\delta\,\varepsilon_{\mathrm{inner}}}}_{\text{(ii) inner-loop}} \;+\; \underbrace{\bigl(L_{\mathrm{fun}}+L_{\mathcal{T}}\bigr)\delta}_{\text{(iii) sampling geometry}} \;+\; \underbrace{L_\phi W_{\max} R_{\mathcal{X}}\,\epsilon_{\mathrm{retr}}}_{\text{(iv) retraction}} \;+\; \underbrace{\epsilon_{\mathrm{mm}}}_{\text{(v) min-max}},\;}
\end{equation}
where $\epsilon_{\mathrm{retr}} := \sup_{\widehat{S} \in \mathcal{S}^{\mathrm{cpt}}_{2r}} \sqrt{\sum_{i=r+1}^{2r}\sigma_i^2(\widehat{S})}$ over the compact set of admissible cores (nonempty by Assumption~\ref{ass:4}(a)), and $\epsilon_{\mathrm{mm}}$ is the min-max tolerance of Assumption~\ref{ass:1}(h), incurred once per target evaluation; being a deterministic uniform bound $\left\lVert \widehat{g}-g \right\rVert_\infty\leq\epsilon_{\mathrm{mm}}$ on the computed target (see the Synthesis in Appendix~\ref{app:thm2}), it enters directly and is not amplified by the bridge factor $\sqrt{N_\delta}$. That is, $\left\lVert \mathcal{A}V^{(k)} - \mathcal{T} V^{(k)} \right\rVert_\infty \leq \varepsilon_{\mathrm{total}}$.
\end{lemma}

\begin{proof}[Sketch]
Decompose $\left\lVert f - g \right\rVert_\infty \leq \left\lVert f - \widehat{f} \right\rVert_\infty + \left\lVert \widehat{f} - \widehat{g} \right\rVert_\infty + \left\lVert \widehat{g} - g \right\rVert_\infty$, where $\widehat{g}$ is the numerically computed target: the middle term is bounded by the Bridge Lemma~\ref{lem:bridge}, giving sources (i)--(iii); the first is the retraction error $\left\lVert f - \widehat{f} \right\rVert_\infty \leq L_\phi \left\lVert w_{\mathrm{out}} \right\rVert_2\, \epsilon_{\mathrm{retr}}\, R_{\mathcal{X}}$, which is (iv); and the last is at most $\epsilon_{\mathrm{mm}}$ by Assumption~\ref{ass:1}(h), which is (v). Details in Appendix~\ref{app:thm2}.
\end{proof}

\begin{theorem}[Convergence of DLR-NE]\label{thm:2}
Let $\{V^{(k)}\}_{k\geq 0}$ be generated by Algorithm~\ref{alg:dlrne}. Under Assumptions~\ref{ass:1}, \ref{ass:2p}, and \ref{ass:4},
\begin{equation}\label{eq:thm2}
\left\lVert V^{(k)} - V^* \right\rVert_\infty \;\leq\; \gamma^k \left\lVert V^{(0)} - V^* \right\rVert_\infty \;+\; \frac{\varepsilon_{\mathrm{total}}}{1-\gamma},
\end{equation}
with $\varepsilon_{\mathrm{total}}$ given by \eqref{eq:lem2}; hence $\limsup_{k\to\infty}\left\lVert V^{(k)} - V^* \right\rVert_\infty \leq \varepsilon_{\mathrm{total}}/(1-\gamma)$.
\end{theorem}

\begin{proof}[Sketch]
By Lemma~\ref{lem:2} and the contraction of $\mathcal{T}$ (Proposition~\ref{prop:contraction}), $\left\lVert V^{(k+1)} - V^* \right\rVert_\infty \leq \varepsilon_{\mathrm{total}} + \gamma\left\lVert V^{(k)} - V^* \right\rVert_\infty$; iterating yields \eqref{eq:thm2}. The full proof, including the covering-number balance between $\delta$ and $N_\delta$, is in Appendix~\ref{app:thm2}.
\end{proof}

Theorem~\ref{thm:2} follows the classical approximate-value-iteration error-propagation framework \citep{munos2003,farahmand2010} and is a worst-case certificate, stated with the ambient covering number $N_\delta$ because the augmented basis itself evolves during iteration; once the basis has settled, $N_\delta$ may equivalently be replaced by the covering number of the effective state set $\{\widehat{V}^\top x : x\in\mathcal{X}\}\subset\mathbb{R}^{2r}$: every function in the fixed-basis class $\mathcal{H}^{(k)}$ depends on $x$ only through the $2r$-dimensional feature $\widehat{V}^\top x$, so the class is in bijection with a function class on that set, and uniform convergence over it requires covering only $\{\widehat{V}^\top x\}$ rather than $\mathcal{X}$, whose cardinality grows as $\delta^{-2r}$; the ambient exponent $n$ enters only while the basis itself drifts, which is why the certificate is stated conservatively with $n$. The steady-state budget of \eqref{eq:lem2} decomposes into five sources, each with a practical control knob: the basis--subspace misalignment~(i) shrinks as the physical low-rank coupling strengthens; the inner-loop price~(ii) by longer inner runs; the sampling cost~(iii) is worst-case $\delta^{-n}$ yet controlled in practice by Latin hypercube sampling at constant $N_b$ on the $2r$-intrinsic-dimensional learned representation (Section~\ref{sec:5}); the retraction tail~(iv) through the rank schedule; and the min-max tolerance~(v) below any prescribed level by refining the $\mathcal{A}\times\mathcal{D}$ discretization at $\mathcal{O}(r^3)$ cost, independent of $n$ (Assumption~\ref{ass:1}(h)).

\subsection{Computational complexity}

\begin{theorem}[Per-iteration complexity and speedup]\label{thm:3}
Under Assumptions~\ref{ass:1} and \ref{ass:2p}, suppose $m = \Theta(n)$ and the inner-loop batch size is constant, $B = \Theta(1)$---justified because on the fixed augmented basis the inner loop optimizes only the $2r\times 2r$ core with $\Theta(r^2)$ effective degrees of freedom, decoupled from $n$, while the covering-set size is a statistical parameter already quantified in $\varepsilon_{\mathrm{total}}$ and not counted again in the flop order. Then a single outer iteration of Algorithm~\ref{alg:dlrne} costs
\begin{equation}\label{eq:thm3-cost}
\mathcal{C}_{\mathrm{DLR}}(n,r,s^*) = \mathcal{O}\bigl(nr^2 + r^3 + s^* nr\bigr),
\end{equation}
whereas full-rank FC-NN value iteration at the same batch setting costs $\mathcal{C}_{\mathrm{full}}(n) = \mathcal{O}(n^2)$. Hence for $r \ll n$ and $s^* = \mathcal{O}(r)$,
\begin{equation}\label{eq:thm3-speedup}
\boxed{\;\frac{\mathcal{C}_{\mathrm{full}}(n)}{\mathcal{C}_{\mathrm{DLR}}(n,r,s^*)} = \Theta\!\left(\frac{n}{r^2}\right).\;}
\end{equation}
If the full-rank baseline additionally performs periodic dense SVD for spectral diagnostics---common in ICS security monitoring---its cost rises to $\mathcal{O}(n^3)$ and the speedup improves to $\Theta(n^2/r^2)$.
\end{theorem}

The complete flop-by-flop derivation is given in Appendix~\ref{app:thm3}; we record three structural remarks. First, storage shrinks from $\Theta(n^2)$ to $\Theta(nr)$ parameters---about $98\%$ reduction at $n=10^3$, $r=10$---critical for online deployment on memory-constrained ICS edge controllers. Second, the hypothesis $s^*=\mathcal{O}(r)$ is a parameter choice, not a structural restriction: the inner problem optimizes only the $2r\times 2r$ core with $\Theta(r^2)$ effective degrees of freedom, so fixed-step gradient descent under the coercive regularizer reaches the accuracy $\varepsilon_{\mathrm{inner}}$ of Assumption~\ref{ass:4}(b) within $s^*=\mathcal{O}\bigl(r^2\log(1/\varepsilon_{\mathrm{inner}})\bigr)$ steps in the worst case, and $\mathcal{O}(r)$ steps suffice once the basis is warm-started; larger $s^*$ degrades the speedup gracefully. Third, unlike supervised DLRA, the min-max of line~\ref{line:extract} would destroy the low-rank advantage if executed in high dimension; confining policy extraction to the coefficient space and solving it through the discretized LP of Assumption~\ref{ass:1}(h) keeps its cost at $\mathcal{O}(r^3)$---in the granularity of the action-space discretization and independent of $n$---so the game extension preserves the $\mathcal{O}(nr^2+r^3)$ order nontrivially. If instead $Q_x$ is a general non-convex--concave function solved without discretization, no polynomial-time global saddle solver exists, and the $r^3$ term should be read as a solver-dependent $C_{\mathrm{ext}}(r_a,r_d,\epsilon_{\mathrm{mm}})$ that remains independent of $n$.

\subsection{Robustness via condition-number control}

The third certificate concerns the physical sensitivity of the induced defense strategy. The inner loop trains the augmented core under the objective $\widehat{\mathcal{L}}_{\mathrm{raw}} + \beta\,\mathcal{R}(\widehat{S})$ with the \emph{spectral regularizer} \citep[adapted from][]{DBLP:conf/nips/SchotthoferYS25}
\begin{equation}\label{eq:reg}
\mathcal{R}(S) := \left\lVert S^\top S - \alpha_S^2 I \right\rVert_F, \qquad \alpha_S^2 := \tfrac{1}{r}\left\lVert S \right\rVert_F^2,
\end{equation}
i.e., $\tfrac{1}{r}\mathcal{R}(S)^2$ is the variance of the squared singular values, unitarily invariant and penalizing spectral anisotropy. The following two propositions, proved in Appendix~\ref{app:props}, make it usable.

\begin{proposition}[Closed-form gradient]\label{prop:grad}
Wherever $\mathcal{R}(S) > 0$, $\nabla \mathcal{R}(S) = 2S(S^\top S - \alpha_S^2 I)/\mathcal{R}(S)$, computable in $\mathcal{O}(r^3)$ flops without ever forming the full $m\times n$ weight.
\end{proposition}

\begin{proposition}[Condition-number bound]\label{prop:cond}
For any $S$ with $\sigma_r(S)>0$,
\[
\kappa(S) \;\leq\; \exp\!\bigl(\mathcal{R}(S)/(\sqrt{2}\,\sigma_r(S)^2)\bigr),
\]
a bound monotone in $\mathcal{R}(S)$ at fixed $\sigma_r(S)$; increasing the weight $\beta$ on $\mathcal{R}(S)$ therefore favors a smaller $\kappa(S) = \sigma_{\max}(S)/\sigma_{\min}(S)$, although $\sigma_r(S)$ itself may shift during training and the proposition alone does not assert pointwise monotonicity in $\beta$.
\end{proposition}

\begin{assumption}[Policy regularity for robustness]\label{ass:5}
\quad
\begin{itemize}[leftmargin=1.4em,itemsep=0.1em]
\item[(a)] For every fixed $(x,a)$, the payoff $Q_x(\cdot,d)$ is $\mu$-strongly concave in $d$, with modulus $\mu>0$ uniform in $x$;
\item[(b)] the reduced reward $r_0$ of Assumption~\ref{ass:1}(f) is $L_{r,a}$- and $L_{r,d}$-Lipschitz in its second and third arguments, respectively;
\item[(c)] the activation $\phi$ is $L_\phi$-Lipschitz and twice continuously differentiable with $\sup_z|\phi''(z)|\leq C_\phi''<\infty$;
\item[(d)] for every $V\in\mathcal{C}(\mathcal{X})$, the expected next-state value $\mathbb{E}_{x'\sim P(\cdot|x,a,d)}[V(x')]$ is jointly continuously differentiable in $(a,d)$, with Jacobian controlled by the Lipschitz constant of $V$ and system constants.
\end{itemize}
\end{assumption}

\begin{corollary}[Game-theoretic robustness]\label{cor:1}
In addition to Assumptions~\ref{ass:1}, \ref{ass:2p}, \ref{ass:4}(a), and \ref{ass:5}, suppose the MTD policy is defined by deterministic greedy extraction
\begin{equation}\label{eq:cor1-policy}
\pi_D(x; \theta) := \arg\max_{d \in \mathcal{D}}\;\min_{a \in \mathcal{A}}\; Q(x,a,d; \theta), \qquad Q(x,a,d;\theta) = r(x,a,d) + \gamma\,\mathbb{E}_{x'}\!\left[w_{\mathrm{out}}^\top \phi(USV^\top x')\right],
\end{equation}
whose unique solution exists by the $\mu$-strong concavity in Assumption~\ref{ass:5}(a). Fix $\pi_A$, and let a Frobenius-bounded perturbation $\Delta S$ of the core induce $\tilde{\theta}$ (with $\tilde{W} = U(S+\Delta S)V^\top$) and $\tilde{\pi}_D$. Then
\begin{equation}\label{eq:cor1-bound}
\boxed{\;\left\lVert V(\cdot;\pi_D,\pi_A) - V(\cdot;\tilde{\pi}_D,\pi_A) \right\rVert_\infty \;\leq\; L_\kappa\,\left\lVert \Delta S \right\rVert_F,\qquad L_\kappa := C_{\mathrm{robust}} \cdot \kappa(S),\;}
\end{equation}
where $C_{\mathrm{robust}} > 0$ depends only on system parameters; and by Proposition~\ref{prop:cond},
\begin{equation}\label{eq:cor1-chain}
L_\kappa \;\leq\; C_{\mathrm{robust}}\exp\!\left(\frac{\mathcal{R}(S)}{\sqrt{2}\,\sigma_r(S)^2}\right).
\end{equation}
\end{corollary}

\begin{proof}
See Appendix~\ref{app:cor1}. For ReLU activations the derivative in the proof is understood in the Clarke generalized-Jacobian sense \citep{clarke1983}, leaving \eqref{eq:cor1-bound}--\eqref{eq:cor1-chain} unchanged.
\end{proof}

Three remarks complete the robustness chain. The certificate is deliberately scoped: it covers only Frobenius-bounded core perturbations with the factors $(U,V,w_{\mathrm{out}})$ and the attacker policy $\pi_A$ held fixed; factor perturbations, observation or state attacks, plant uncertainty, and the full equilibrium sensitivity under attacker re-optimization are left to future work. First, within this scope, Corollary~\ref{cor:1} translates \emph{adversarial robustness} into \emph{game-theoretic robustness}: the APT's goal is not to perturb observations but to compromise the \emph{execution fidelity} of the MTD policy, and even a small $\Delta S$ injected into the core causes an equilibrium-value loss explicitly governed by $\kappa(S)$. Second, the bound is an a priori worst-case certificate: it guarantees that a finite Lipschitz constant exists and is tunable through $\beta$, without predicting that the constant is attained---Section~\ref{sec:5} shows that the measured sensitivity stays well below the certificate. Its value therefore lies in the existence guarantee and in the tuning direction it prescribes (larger $\beta$ reduces $\mathcal{R}(S)$, hence the certificate, and is observed to reduce $\kappa(S)$ in our experiments), rather than in the precise prediction of measured sensitivities. Third, the inequality yields a \emph{computable robustness budget}: given a tolerable loss $\Delta V_{\max}$ and anticipated attack strength $\left\lVert \Delta S \right\rVert_{\max}$, the defender chooses $\beta$ so that $C_{\mathrm{robust}}\exp\bigl(\mathcal{R}(S)/(\sqrt{2}\sigma_r(S)^2)\bigr)\left\lVert \Delta S \right\rVert_{\max} \leq \Delta V_{\max}$, meeting the physical safety constraint while retaining computational efficiency; if $(\pi_D^*,\pi_A^*)$ is a Nash pair, the same bound quantifies the utility loss from equilibrium deviation, tightest at $\kappa(S)=1$.

% ============================================================
% SECTION 5: NUMERICAL EXPERIMENTS
% ============================================================
\section{Numerical experiments}\label{sec:5}

This section validates the theory of Sections~\ref{sec:3} and~\ref{sec:4} through six numerical experiments, each targeting one \emph{specific, verifiable prediction}; since $V^*$ of the continuous-state game admits no closed form, the converged full-rank network serves as a proxy benchmark, and its intrinsic approximation gap $\epsilon_{\mathrm{NN}}$ is reported as a reference floor. Experiments 1--3 verify Theorems~\ref{thm:1}--\ref{thm:3} on approximation, convergence, and complexity; Experiment 4 verifies the robustness certificate of Corollary~\ref{cor:1}; Experiment 5 quantifies the resulting compression--accuracy trade-off; and Experiment 6 ablates the two algorithmic mechanisms. All experiments are conducted on a nonlinear power-system testbed. Table~\ref{tab:summary} summarizes the theoretical prediction and the measured result of each experiment. All source code is available at \url{https://github.com/tz98lab/dlr-ne-ics-security}.

\begin{table}[t]
\centering
\caption{Summary of the six experiments: theoretical prediction versus measured result.}
\label{tab:summary}
\small
\begin{tabular}{@{}c p{5.9cm} p{6.7cm}@{}}
\toprule
Exp. & Theoretical prediction & Measured result \\
\midrule
1 & Thm.~\ref{thm:1}: truncation error bounded by the EYM certificate; ICS coupling is low-rank & error $3.3\times10^{-2}\!\to\!2.0\times10^{-3}$ for $r{=}5\!\to\!100$, $>2$ orders below the certificate; spectrum cliffs at $i\approx202\approx n/5$ \\
2 & Thm.~\ref{thm:2}: geometric decay $\gamma^k$ then a plateau set by five error sources & two phases on all configurations; $s^*$ dominates the plateau ($20\times$); a priori bound $3\times10^3$ vs.\ measured plateau $\approx2$ (conservative) \\
3 & Thm.~\ref{thm:3}: per-step $\mathcal{O}(nr^2)$, $\Theta(n/r^2)$ speedup & log--log slopes consistent with $\mathcal{O}(nr^2)$ vs.\ $\mathcal{O}(n^2)$; FLOPs speedup $\approx12\times$ at $n=2000$ \\
4 & Cor.~\ref{cor:1}: linear sensitivity with a $\beta$-tunable constant & per-network linear, $R^2\in[0.87,0.98]$, $L_\kappa\in[0.09,0.51]$; $\kappa(S)$: $83.7\to31.0$ as $\beta:0\to1$ \\
5 & Thms.~\ref{thm:1}+\ref{thm:3}: convex compression--accuracy frontier & $r{=}5$: $94\%$ parameter cut at $2.3\%$ utility loss \\
6 & Lem.~\ref{lem:1} + Assump.~\ref{ass:4}(c): both mechanisms necessary & Fixed-Basis stagnates; No-Retraction destabilizes late; full algorithm stable at a low plateau \\
\bottomrule
\end{tabular}
\end{table}

\subsection{Experimental setup}

\paragraph{Testbed.}
Power grids are a canonical instance of ICS: bus voltages and frequencies are regulated by supervisory control loops (SCADA/EMS) over geographically distributed physical dynamics, and they constitute the standard testbed for ICS security studies \citep{etxezarreta2026sdicsmtd,hu2021mtdpower,wang2025mtdsmartgrid}. The testbed is a nonlinear power-system model in the form of \eqref{eq:dyn}--\eqref{eq:decomp}: $n$ bus voltages evolve by an explicit Euler discretization of AC power-flow coupling with saturation to $[0.95,1.05]^n$, state-dependent low-rank attack and defense influence matrices $B_a(x)=\mathrm{diag}(\sigma(x-0.5))\bar B_a$, $B_d(x)=\bar B_d\,\mathrm{softmax}(W_dx)$ with $r_a=3$, $r_d=5$, and stage reward $r(x,a,d)=-\|C_r x\|_2^2-c_a\|a\|_2^2-c_d\|d\|_2^2$ ($r_r=5$). Scales $n\in\{100,500,1000,2000\}$ are realized with random geometric networks and IEEE 118/300-bus topologies. States are sampled by Latin hypercube sampling; the residual sampling effect is absorbed into $\varepsilon_{\mathrm{inner}}$ of Theorem~\ref{thm:2}.

\paragraph{Baselines.}
We compare against three baselines. FC-NN-VI: full-rank value iteration with the same network width, the exact procedure DLR-NE compresses. FC-NN-AVI: its approximate variant at identical sampling settings. Fixed-Basis: an ablation that disables basis augmentation and keeps $(U,V)$ frozen at their initial values. No-Retraction: an ablation that skips the truncated-SVD retraction, so the iterate remains on the rank-$2r$ augmented manifold.

\paragraph{Protocol.}
All experiments use Python/PyTorch on an Intel Xeon Gold 6248R. Each configuration is run with five random seeds and we report mean $\pm$ std, unless stated otherwise; Experiment~1 additionally reports a single representative training run, as its truncation error is evaluated on $10^4$ test states and is dominated by the spectral decay rather than by initialization. FLOPs are counted per Theorem~\ref{thm:3} as the primary complexity metric, with wall-clock time as secondary validation. The inner problems are solved by Adam \citep{kingma2015adam} with learning rate $0.01$; the complete hyperparameter list is deferred to Appendix~\ref{app:hyper}. Termination of Algorithm~\ref{alg:dlrne} is declared on the current batch when the residual $\max_i \left\lvert V(x_i;\theta^{(k)}) - y_i \right\rvert$ falls below $\epsilon_{\mathrm{tol}}$; by Theorem~\ref{thm:2}, such a batch residual is representative of the full-space error up to the sampling-geometry term (iii), the standard approximate-value-iteration practice.

\subsection{Results and analysis}

\begin{figure}[!t]
\centering
\includegraphics[width=0.95\textwidth,height=0.62\textheight,keepaspectratio]{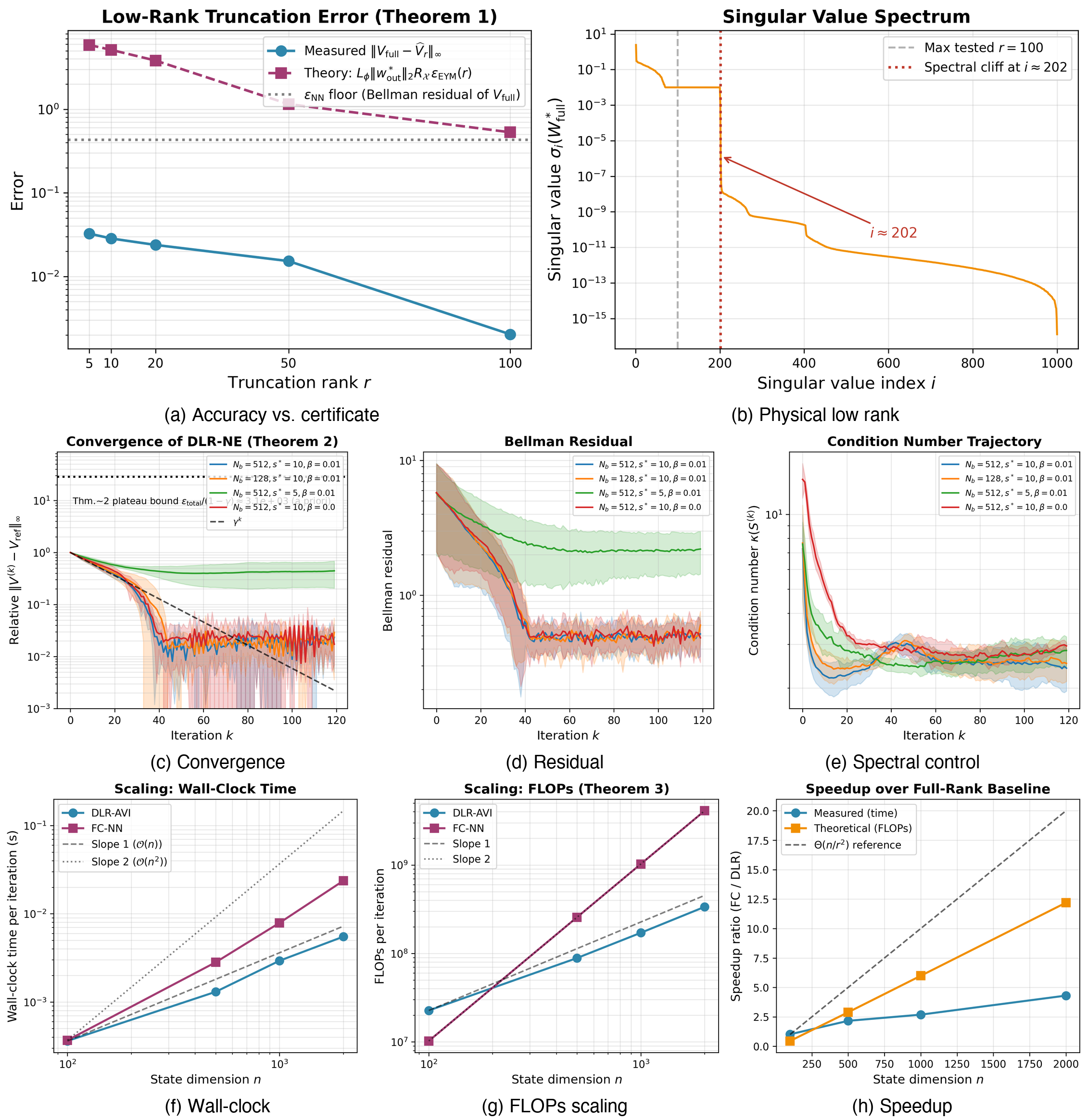}
\caption{Experiments 1--3 at $n=m=1000$. (a),(b) Truncation error and singular-value spectrum of the trained full-rank benchmark: the error sits two orders of magnitude below its certificate, and the spectrum exhibits a sharp cliff at $i\approx 202\approx n/5$, evidencing the physical low-rank structure of ICS coupling. (c)--(e) Convergence of DLR-NE: geometric decay along the $\gamma^k$ envelope into a plateau governed by the inner-loop budget $s^*$, with the a priori bound three orders of magnitude conservative. (f)--(h) Complexity scaling: fitted slopes are consistent with $\mathcal{O}(nr^2)$ against $\mathcal{O}(n^2)$, and the speedup grows linearly with $n$, reaching $\approx12\times$ in FLOPs at $n=2000$.}
\label{fig:exp123}
\end{figure}

\paragraph{Experiment 1 (truncation error vs.\ Theorem~\ref{thm:1}).}
Fixing $n=m=1000$, we train the full-rank benchmark with spectral regularization (Appendix~\ref{app:hyper}), then compress it by truncated SVD of its input-layer weight to rank $r\in\{5,10,20,50,100\}$, without any retraining. In Fig.~\ref{fig:exp123}(a), the horizontal axis is the truncation rank $r$ and the vertical axis is the worst-case approximation error $\|V_{\mathrm{full}}-\widehat{V}_r\|_\infty$ over $10^4$ test states; the blue curve is the measured error, the dashed curve is the certificate of Theorem~\ref{thm:1}, and the dotted line is the $\epsilon_{\mathrm{NN}}$ floor of the full-rank benchmark. The measured error decays steadily from $\approx 3.3\times10^{-2}$ at $r=5$ to $\approx 2.0\times10^{-3}$ at $r=100$, staying over two orders of magnitude below the certificate and even below the floor. Fig.~\ref{fig:exp123}(b) plots the singular-value spectrum $\sigma_i(W_{\mathrm{full}}^*)$ of the trained weight, with the singular-value index on the horizontal axis and $\sigma_i$ on a logarithmic vertical axis: the spectrum remains flat up to a sharp cliff at $i\approx 202$ and drops to $10^{-10}$ immediately beyond.

The message is twofold. First, Theorem~\ref{thm:1} is validated with a wide safety margin: even the largest tested truncation loses only $\approx 3\times10^{-2}$ of accuracy, two orders of magnitude below its own certificate. Second, truncation is not a compromise but a free lunch---the measured error falls even below the $\epsilon_{\mathrm{NN}}$ floor ($\approx 4.3\times10^{-1}$), because discarding the noise-dominated tail of the fitted weight yields a value function \emph{more accurate} than the full-rank benchmark, confirming the theory with a large margin. The spectrum panel shows that the full-rank network has learned an $\approx n/5$-dimensional representation of the game: of its 1000 singular directions, only the first $\approx 202$ carry non-negligible energy. This elevates Experiment~1 from validating a bound to directly observing the physical low-rank structure of ICS coupling---few attack and defense channels---in the trained result. It also explains panel~(a): every tested rank $r\le 100$ lies in the flat region of the spectrum, so the measured cost of discarding the entire tail beyond $r$ is at most $\approx 2\times10^{-3}$; the part beyond the cliff at $i\approx 202$ is pure numerical noise.

\paragraph{Experiment 2 (convergence vs.\ Theorem~\ref{thm:2}).}
At $n=1000$, $r=20$, $\gamma=0.95$, we run DLR-NE for $120$ iterations under four hyperparameter configurations, each with five seeds. In Fig.~\ref{fig:exp123}(c), the horizontal axis is the iteration $k$ and the vertical axis is the distance to the pretrained full-rank benchmark relative to its initial value (solid lines: mean; shaded bands: $\pm1$ std): every configuration first decays along the dashed $\gamma^k$ envelope and then bends into a plateau---the two phases predicted by Theorem~\ref{thm:2}. The inner budget dominates the plateau height: $s^*=5$ (green) stalls at $0.44\pm0.24$, about $20\times$ higher than $s^*=10$ (blue, $0.024\pm0.018$); halving the batch to $N_b=128$ (orange) leaves the plateau essentially unchanged ($0.029\pm0.026$); and $\beta=0$ (red) attains the lowest error ($0.017\pm0.010$). The dotted line marks the a priori plateau bound $\varepsilon_{\mathrm{total}}/(1-\gamma)\approx 3\times10^3$ estimated from the five error sources of Lemma~\ref{lem:2}. Fig.~\ref{fig:exp123}(d) shows the Bellman residual decaying in parallel and plateauing consistently with panel~(c), and Fig.~\ref{fig:exp123}(e) shows the condition-number trajectories ordered by $\beta$, with $\beta=0$ settling highest ($\kappa\approx 3.0$ vs.\ $2.4$).

The plateau responds to the knobs in the directions the five sources predict. Source~(ii) is dominant: an under-solved inner problem leaves a persistent per-iteration error that accumulates into the plateau, so halving $s^*$ costs a factor of $\approx 20$ in accuracy. Source~(iii) is comparatively negligible: sweeping $N_b$ from $16$ to $512$ moves the plateau only within seed noise ($0.032$--$0.024$), because a few hundred samples already cover the intrinsic $2r$-dimensional subspace---an appealing sample-efficiency property. The weight $\beta$ leaves accuracy almost untouched but orders the condition number, foreshadowing Corollary~\ref{cor:1}: spectral regularization buys a better-conditioned core at almost no cost. Finally, the a priori bound overestimates the measured plateau by roughly three orders of magnitude: Theorem~\ref{thm:2} is validated as a conservative certificate, in the same sense as Theorem~\ref{thm:1} in Experiment~1.

\paragraph{Experiment 3 (complexity vs.\ Theorem~\ref{thm:3}).}
We compare DLR-NE with the full-rank baseline over $n\in\{100,500,1000,2000\}$ at $r=10$. In Fig.~\ref{fig:exp123}(f), the horizontal axis is the state dimension $n$ and the vertical axis is the wall-clock time per iteration (both log-scaled): the full-rank cost grows steeply while DLR-NE grows mildly. Fig.~\ref{fig:exp123}(g) reports FLOPs per iteration on the same axes, whose fitted slopes are consistent with $\mathcal{O}(nr^2)$ for DLR-NE against $\mathcal{O}(n^2)$ for the baseline. Fig.~\ref{fig:exp123}(h) summarizes the speedup ratio (FC/DLR): the FLOPs speedup (orange) grows linearly with $n$, reaching $\approx12\times$ at $n=2000$ toward the asymptotic reference $n/r^2=20$, while the measured wall-clock speedup (blue) reaches $\approx4\times$.

The gap between the $\approx12\times$ FLOPs and $\approx4\times$ wall-clock speedups is expected at these scales: neither implementation has entered the asymptotic regime, and wall-clock is further limited by constant factors and cache effects, which motivates our FLOPs-first reporting. Against a baseline that additionally performs dense SVD for spectral diagnostics, the FLOPs advantage exceeds two orders of magnitude. At $n=100$ the fixed overhead of basis augmentation and core SVD is not yet amortized---the FLOPs ratio is even below one---and the advantage only turns positive for $n\gtrsim500$; from there, by the $\Theta(n/r^2)$ law, the advantage widens linearly with $n$ without bound---extrapolating to industrial dimensions where full-rank baselines become memory-bound.

\paragraph{Experiment 4 (robustness vs.\ Corollary~\ref{cor:1}).}
We test the certificate of Corollary~\ref{cor:1} in two regimes. First, sweeping $\beta\in\{0,0.01,0.1,1\}$ confirms the regularization link of Proposition~\ref{prop:cond}: the converged core's condition number drops from $\kappa(S)=83.7$ at $\beta=0$ to $31.0$ at $\beta=1$, with fitted sensitivity slopes decreasing in $\beta$. Second, fixing $\beta=0.1$, we train eight networks ($r\in\{10,20,40,60\}$, two seeds) and perturb each converged core. In Fig.~\ref{fig:exp456}(a), the horizontal axis is the core perturbation $\|\Delta S\|_F$ and the vertical axis is the resulting value deviation: all eight networks respond linearly ($R^2\in[0.87,0.98]$) with fitted slopes confined to a narrow band $[0.09,0.51]$, so the bound is never violated. Fig.~\ref{fig:exp456}(b) plots the fitted slope against the converged $\kappa(S)$ on log--log axes, with a slope-one proportionality reference (dashed).

Two findings follow. Per network, the certificate holds: every network obeys linear sensitivity with a constant far below the a priori bound. Across the family, however, $L_\kappa$ does not track $\kappa(S)$ ($R^2=0.17$) even though $\kappa(S)$ spans two orders of magnitude ($47$ to $1583$): the proportionality form of Corollary~\ref{cor:1} is an a priori worst-case certificate---it guarantees a finite, $\beta$-tunable sensitivity constant rather than predicting its value. This conservative-certificate behavior is consistent with Experiments~1 and~2.

We close with a remark on magnitudes. Since $\kappa(S)=\sigma_{\max}(S)/\sigma_{\min}(S)$ and spectral regularization floors $\sigma_{\min}(S)$ near a small threshold, $\kappa(S)$ is a tail-sensitive quantity: minor run-to-run differences in the floor are amplified into orders-of-magnitude differences in the ratio. Fig.~\ref{fig:exp123}(e) reported the trajectory of one default configuration ($r=20$, five seeds, batch-averaged), settling at $\kappa\approx2.4$--$3.0$, whereas the present experiment reports per-network converged values across four ranks and two seeds; the two sets of numbers therefore answer different questions and neither is ``the'' typical condition number. What is reproducible across both is the object the theory actually requires: $\kappa(S)$ decreases monotonically in $\beta$, and the certified sensitivity bound is never violated.

\begin{figure}[!htbp]
\centering
\includegraphics[width=0.8\textwidth,height=0.5\textheight,keepaspectratio]{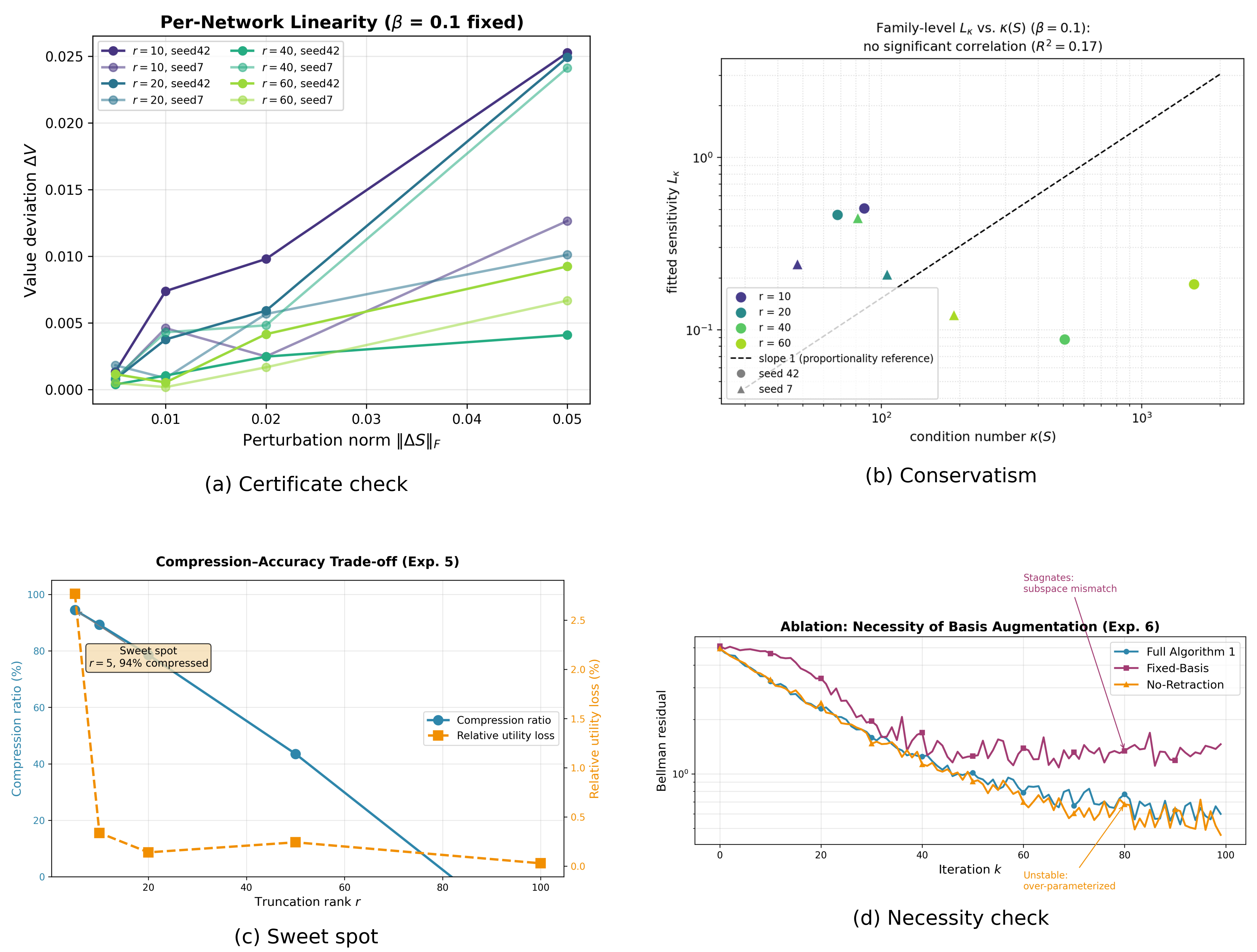}
\caption{Experiments 4--6 at $n=m=1000$. (a),(b) Robustness: value deviation is linear in the core perturbation in all eight networks, while the fitted sensitivity does not track the condition number---the certificate of Corollary~\ref{cor:1} is valid but conservative. (c) Compression--accuracy trade-off: $r=5$ removes $\approx94\%$ of the parameters at $\approx2.3\%$ utility loss. (d) Ablation: disabling basis augmentation stagnates convergence and skipping retraction destabilizes it---both mechanisms are necessary.}
\label{fig:exp456}
\end{figure}

\paragraph{Experiment 5 (compression--accuracy trade-off).}
We vary $r\in\{5,10,20,50,100\}$ at $n=1000$. In Fig.~\ref{fig:exp456}(c), the horizontal axis is the truncation rank $r$, the left vertical axis is the fraction of removed parameters, and the right vertical axis is the relative utility loss: the frontier is sharply convex---utility recovers most of its full-rank value at very small rank and saturates, while compression improves steeply.

Quantitatively, the sweet spot $r=5$ removes $\approx94\%$ of the parameters at $\approx2.3\%$ utility loss. This Pareto behavior is jointly predicted by Theorems~\ref{thm:1} and~\ref{thm:3}---the truncation error decays with the spectral tail while the cost saving grows with $n/r^2$---and quantifies the central claim that the physical low-rank coupling of ICS dynamics permits tiny truncation ranks in practice.

\paragraph{Experiment 6 (ablation: necessity of both mechanisms).}
At $n=1000$, $r=10$, Fig.~\ref{fig:exp456}(d) tracks the Bellman residual of three variants; the horizontal axis is the iteration $k$ and the vertical axis is the Bellman residual. Fixed-Basis (basis augmentation disabled, $(U,V)$ frozen) decays initially but stagnates at a high plateau; No-Retraction (the truncated-SVD retraction skipped, the iterate staying on the rank-$2r$ augmented manifold) decays fast early yet destabilizes in the later phase; only the full algorithm combines fast initial decay with a stable low plateau.

The two failures pinpoint the role of each mechanism. A frozen basis cannot align with the state-dependent local subspaces certified by Lemma~\ref{lem:1}, so the basis-subspace source~(i) dominates the Fixed-Basis plateau; without retraction the rank-$2r$ iterate drifts on the augmented manifold, and the missing anchoring eventually amplifies the inner and retraction sources. The full algorithm suppresses both failure modes, corroborating Lemma~\ref{lem:1} and Assumption~\ref{ass:4}(c).

% ============================================================
% SECTION 6: CONCLUSION
% ============================================================
\section{Conclusion}\label{sec:6}

This paper addressed equilibrium computation for the strategic interaction between advanced persistent threats and moving target defense in industrial control systems, modeled as a zero-sum stochastic game in a high-dimensional continuous state space. While game-theoretic formulations of this interaction are well established, their practical value is decided by two questions that existing theory leaves open: whether the equilibrium value can be computed at industrial state dimensions in real time, and whether the resulting defense strategy remains safe under adversarial perturbation of its own parameters.

Our answer rests on a structural observation and a certified algorithm. Structurally, we proved that the low-rank coupling of industrial dynamics---few attack entry points and limited reconfigurable resources---propagates through the non-smooth Bellman operator: the augmented gradient matrix of Lemma~\ref{lem:1} certifies that every Bellman target is locally captured by a $2r$-dimensional subspace, and Theorem~\ref{thm:1} turns the locally certified structure into an explicit global approximation bound governed by the tail singular values of the full-rank benchmark, whose fast decay for ICS dynamics is evidenced by the spectral measurements of Section~\ref{sec:5}. Algorithmically, since these certified subspaces drift along the value-iteration trajectory, the DLR-NE algorithm augments the rank-$r$ search space each iteration, regularizes the core-matrix spectrum, and retracts via truncated SVD, extracting a Nash equilibrium at every step. Three guarantees accompany the algorithm: geometric convergence to a neighborhood whose error decomposes into five physically interpretable sources (Theorem~\ref{thm:2}), per-step cost $\mathcal{O}(nr^2)$ yielding a $\Theta(n/r^2)$ speedup (Theorem~\ref{thm:3}), and an explicit robustness bound in which a single weight trades accuracy against certified safety through condition-number control (Corollary~\ref{cor:1}). Six numerical experiments on a nonlinear power-system testbed confirm each theoretical prediction individually; at the Pareto sweet spot, the method attains $94\%$ parameter compression at approximately $2.3\%$ relative utility loss.

The certificates are honest about two limitations, one structural and one computational. First, the theory is developed for the simultaneous-move, zero-sum abstraction with full state observability, which idealizes the sequential, information-asymmetric nature of real APT campaigns. Second---and this limitation is confirmed by both theory and experiment---the low-rank advantage has a scale threshold: Theorem~\ref{thm:3} predicts a speedup of only $\Theta(n/r^2)$, which is below one for $n<r^2$, and the measurements of Experiment~3 agree, the FLOPs ratio at $n=100$ being smaller than one so that the advantage turns positive only for $n\gtrsim500$; DLR-NE is thus a tool for large-scale ICS instances, and for small systems full-rank value iteration remains the simpler and cheaper choice. Relatedly, the certificates are conservative by construction: the a priori convergence bound overestimates the measured plateau by roughly three orders of magnitude, and the robustness certificate guarantees a finite sensitivity constant without predicting its value.

Two directions follow naturally. On the application side, the certificates established here provide the theoretical basis for a deployment-oriented study on real ICS testbeds beyond the power-grid instance validated here, covering controller synthesis under the computed equilibria and integration with existing intrusion detection pipelines. On the theory side, extending the augmented-gradient analysis to games with partial observability and asymmetric information---where the defender's measurement of the attack surface is itself noisy---would close the gap between the simultaneous-move abstraction and sequential intrusion campaigns.

% ---------- end of §6 ----------

% ============================================================
% APPENDICES (Lemma 1 proof now embedded in Section 3)
% ============================================================
\appendix
\makeatletter
\renewcommand{\@seccntformat}[1]{Appendix \csname the#1\endcsname.\quad}
\makeatother

\section{Proofs of Propositions~\ref{prop:grad} and \ref{prop:cond}}\label{app:props}

\begin{proof}[Proof of Proposition~\ref{prop:grad}]
Write $\mathcal{R}(S) = \left\lVert M \right\rVert_F$ with $M := S^\top S - \alpha_S^2 I_r$ and $\alpha_S^2 = \tfrac{1}{r}\left\lVert S \right\rVert_F^2$. Where $\mathcal{R}(S) > 0$, the Fr\'echet derivative is
\[
\nabla \mathcal{R}(S) = \frac{1}{\mathcal{R}(S)}\,D\!\left(\tfrac{1}{2}\left\lVert M \right\rVert_F^2\right)[S].
\]
For the first term, $D(S^\top S)[H] = H^\top S + S^\top H$; for the second, $D(\alpha_S^2 I)[H] = \tfrac{2}{r}\langle S, H\rangle_F I$. Since $\langle M, \alpha_S^2\text{-term}\rangle_F = \alpha_S^2 \mathrm{tr}(M) = 0$ (because $\mathrm{tr}(S^\top S) = r\alpha_S^2$ by definition of $\alpha_S^2$), the $\alpha_S$-variation drops out of the gradient:
\[
D\!\left(\tfrac{1}{2}\left\lVert M \right\rVert_F^2\right)[H] = \langle M, H^\top S + S^\top H\rangle_F = 2\langle SM, H\rangle_F.
\]
Hence $\nabla\mathcal{R}(S) = 2SM/\mathcal{R}(S)$, which is the gradient formula of Proposition~\ref{prop:grad}. All operations are $r\times r$ matrix products plus one Frobenius norm: $\mathcal{O}(r^3)$ flops.
\end{proof}

\begin{proof}[Proof of Proposition~\ref{prop:cond}]
Let $\sigma_1 \geq \cdots \geq \sigma_r > 0$ be the singular values of $S$. From \eqref{eq:reg},
\[
\mathcal{R}(S)^2 = \sum_{i=1}^{r}\bigl(\sigma_i^2 - \alpha_S^2\bigr)^2, \qquad \alpha_S^2 = \tfrac{1}{r}\sum_{i=1}^{r}\sigma_i^2.
\]
Let $\kappa = \kappa(S) = \sigma_1/\sigma_r$. Using $\log\kappa = \log\sigma_1 - \log\sigma_r = \sum_{i=1}^{r-1}\log\frac{\sigma_i}{\sigma_{i+1}}$ and the elementary inequality $\log t \leq t - 1$,
\[
\log\kappa \;\leq\; \sum_{i=1}^{r-1}\left(\frac{\sigma_i}{\sigma_{i+1}} - 1\right).
\]
Since $\sigma_{i+1}^2 \geq \sigma_r^2$ for all $i$, $\sigma_i/\sigma_{i+1} - 1 \leq (\sigma_i^2 - \sigma_{i+1}^2)/(2\sigma_r^2)$. Summing telescopically,
\[
\log\kappa \;\leq\; \frac{\sigma_1^2 - \sigma_r^2}{2\sigma_r^2} \;\leq\; \frac{1}{2\sigma_r^2}\sum_{i=1}^{r}\left\lvert \sigma_i^2 - \alpha_S^2 \right\rvert \;\leq\; \frac{\sqrt{r}}{2\sigma_r^2}\sqrt{\tfrac{1}{r}\sum_i (\sigma_i^2-\alpha_S^2)^2} \;\leq\; \frac{\mathcal{R}(S)}{\sqrt{2}\,\sigma_r^2},
\]
where the last two steps use Cauchy--Schwarz and $\sqrt{r} \leq r/\sqrt{2}$ for $r \geq 2$ (and the case $r=1$ gives $\kappa=1$ trivially). Exponentiating yields the bound of Proposition~\ref{prop:cond}.
\end{proof}

% ============================================================
% APPENDIX C: Proofs for Section 6.1
% ============================================================
\section{Proofs of Lemma~\ref{lem:bridge}, Lemma~\ref{lem:2}, and Theorem~\ref{thm:2}}\label{app:thm2}

\begin{proof}[Proof of Lemma~\ref{lem:bridge}]
Since $f^* \in \mathcal{H}^{(k)}$ is feasible for the empirical problem and $\widehat{f}$ is its minimizer,
\[
\frac{1}{N_\delta}\sum_{j=1}^{N_\delta}\bigl(\widehat{f}(x_j)-g(x_j)\bigr)^2 \;\leq\; \frac{1}{N_\delta}\sum_{j=1}^{N_\delta}\bigl(f^*(x_j)-g(x_j)\bigr)^2 + \varepsilon_{\mathrm{inner}} \;\leq\; \epsilon_{\mathrm{app}}^2 + \varepsilon_{\mathrm{inner}},
\]
where the last inequality uses $\left\lVert f^*-g \right\rVert_\infty \leq \epsilon_{\mathrm{app}}$ (Assumption~\ref{ass:4}(c)). Hence there exists $j^* \in \{1,\dots,N_\delta\}$ with
\[
\bigl|\widehat{f}(x_{j^*}) - g(x_{j^*})\bigr| \;\leq\; \sqrt{\epsilon_{\mathrm{app}}^2 + \varepsilon_{\mathrm{inner}}} \;\leq\; \epsilon_{\mathrm{app}} + \sqrt{\varepsilon_{\mathrm{inner}}}.
\]
Fix an arbitrary $x \in \mathcal{X}$ and choose $x_j \in \mathcal{C}_\delta$ with $\left\lVert x - x_j \right\rVert_2 \leq \delta$ (possible by the covering property). Then, by the Lipschitz assumptions (ii),
\[
|\widehat{f}(x) - g(x)| \;\leq\; \underbrace{|\widehat{f}(x) - \widehat{f}(x_j)|}_{\leq\, L_{\mathrm{fun}}\,\delta} \;+\; \underbrace{|\widehat{f}(x_j) - g(x_j)|}_{\leq\; \sqrt{N_\delta}\,\sqrt{\frac{1}{N_\delta}\sum_i (\widehat f(x_i)-g(x_i))^2}} \;+\; \underbrace{|g(x_j) - g(x)|}_{\leq\, L_{\mathcal{T}}\,\delta}.
\]
The middle term is bounded by Cauchy--Schwarz:
\[
\max_{j}\,|\widehat{f}(x_j)-g(x_j)| \;\leq\; \sqrt{N_\delta}\cdot\sqrt{\frac{1}{N_\delta}\sum_{j=1}^{N_\delta}\bigl(\widehat{f}(x_j)-g(x_j)\bigr)^2} \;\leq\; \sqrt{N_\delta}\,\epsilon_{\mathrm{app}} + \sqrt{N_\delta\,\varepsilon_{\mathrm{inner}}}.
\]
Combining and taking the supremum over $x \in \mathcal{X}$ yields \eqref{eq:bridge}.
\end{proof}

\begin{proof}[Proof of Lemma~\ref{lem:2}]
We bound the two terms of the decomposition $\left\lVert f - g \right\rVert_\infty \leq \left\lVert f - \widehat{f} \right\rVert_\infty + \left\lVert \widehat{f} - g \right\rVert_\infty$.

\emph{Term 1 (Bridge).} By the inner-loop termination criterion of Algorithm~\ref{alg:dlrne}, $\widehat{f}$ achieves empirical risk at most $\epsilon_{\mathrm{app}}^2 + \varepsilon_{\mathrm{inner}}$ on $\mathcal{C}_\delta$. Assumption~\ref{ass:4}(c) supplies the candidate $f^* \in \mathcal{H}^{(k)}$ with $\left\lVert f^*-g \right\rVert_\infty \leq \epsilon_{\mathrm{app}}$, and the Lipschitz constants $L_{\mathrm{fun}}$ (uniform on $\mathcal{H}^{(k)}$ by the smoothness and boundedness assumptions on the activation, weights, and regularizer) and $L_{\mathcal{T}}$ (derived in Section~2) satisfy the hypotheses of Lemma~\ref{lem:bridge}. Therefore
\[
\left\lVert \widehat{f} - g \right\rVert_\infty \;\leq\; \sqrt{N_\delta}\,\epsilon_{\mathrm{app}} + \sqrt{N_\delta\,\varepsilon_{\mathrm{inner}}} + \bigl(L_{\mathrm{fun}}+L_{\mathcal{T}}\bigr)\delta.
\]

\emph{Term 2 (Retraction).} Write $\widehat{W} = \widehat{U}\widehat{S}\widehat{V}^\top$ and $W = \widehat{U}\,\Pi_{\mathcal{M}_r}(\widehat{S})\,\widehat{V}^\top$, where $\Pi_{\mathcal{M}_r}(\widehat{S})$ is the rank-$r$ truncated SVD of $\widehat{S}$. Since $\widehat{U},\widehat{V}$ have orthonormal columns,
\[
\left\lVert W - \widehat{W} \right\rVert_2 \;=\; \bigl\|\Pi_{\mathcal{M}_r}(\widehat{S}) - \widehat{S}\bigr\|_2 \;\leq\; \sqrt{\sum_{i=r+1}^{2r}\sigma_i^2(\widehat{S})} \;\leq\; \epsilon_{\mathrm{retr}},
\]
where the first inequality is Eckart--Young and $\widehat{S} \in \mathcal{S}^{\mathrm{cpt}}_{2r}$ by the coercivity of the regularizer (Assumption~\ref{ass:4}(a)), which keeps the inner-loop trajectory inside a compact level set determined by the initialization and $\beta$. For any $x \in \mathcal{X}$ with $\left\lVert x \right\rVert_2 \leq R_{\mathcal{X}}$,
\[
|f(x) - \widehat{f}(x)| \;=\; \bigl|w_{\mathrm{out}}^\top\bigl(\phi(Wx) - \phi(\widehat{W}x)\bigr)\bigr| \;\leq\; \left\lVert w_{\mathrm{out}} \right\rVert_2\,L_\phi\,\left\lVert (W-\widehat{W})x \right\rVert_2 \;\leq\; L_\phi\,\left\lVert w_{\mathrm{out}} \right\rVert_2\,\epsilon_{\mathrm{retr}}\,R_{\mathcal{X}},
\]
using the $L_\phi$-Lipschitz activation (Assumption~\ref{ass:5}(c)) and $\left\lVert w_{\mathrm{out}} \right\rVert_2 \leq W_{\max}$.

\emph{Synthesis.} The min-max targets are evaluated only to accuracy $\epsilon_{\mathrm{mm}}$ (Assumption~\ref{ass:1}(h)); writing $\widehat{g}$ for the computed target, $\left\lVert \widehat{g}-g \right\rVert_\infty\leq\epsilon_{\mathrm{mm}}$, and this discrepancy is carried through the decomposition unchanged, contributing source (v) of \eqref{eq:lem2}. The triangle inequality combines the two bounds:
\[
\left\lVert f - g \right\rVert_\infty \;\leq\; \underbrace{\sqrt{N_\delta}\,\epsilon_{\mathrm{app}}}_{\text{(i)}} \;+\; \underbrace{\sqrt{N_\delta\,\varepsilon_{\mathrm{inner}}}}_{\text{(ii)}} \;+\; \underbrace{\bigl(L_{\mathrm{fun}}+L_{\mathcal{T}}\bigr)\delta}_{\text{(iii)}} \;+\; \underbrace{L_\phi W_{\max} R_{\mathcal{X}}\,\epsilon_{\mathrm{retr}}}_{\text{(iv)}} \;+\; \underbrace{\epsilon_{\mathrm{mm}}}_{\text{(v)}} \;=\; \varepsilon_{\mathrm{total}},
\]
which is exactly \eqref{eq:lem2}; the $\sqrt{N_\delta}$ factors in (i)--(ii) are structural and not removable (see the remark following Lemma~\ref{lem:bridge}).
\end{proof}

\begin{proof}[Proof of Theorem~\ref{thm:2}]
Let $\mathcal{A}:\mathcal{F}_r \to \mathcal{F}_r$ denote the deterministic operator realized by one outer iteration of Algorithm~\ref{alg:dlrne} (targets, augmentation, inner loop, retraction, extraction), so that $V^{(k+1)} = \mathcal{A}V^{(k)}$. By Lemma~\ref{lem:2},
\[
\left\lVert \mathcal{A}V^{(k)} - \mathcal{T} V^{(k)} \right\rVert_\infty \;\leq\; \varepsilon_{\mathrm{total}} \qquad \text{for all } k \geq 0.
\]
Since $V^* = \mathcal{T} V^*$ (Proposition~\ref{prop:contraction}), the triangle inequality and the $\gamma$-contraction of $\mathcal{T}$ give
\[
\left\lVert V^{(k+1)} - V^* \right\rVert_\infty \;\leq\; \left\lVert \mathcal{A}V^{(k)} - \mathcal{T} V^{(k)} \right\rVert_\infty + \left\lVert \mathcal{T} V^{(k)} - \mathcal{T} V^* \right\rVert_\infty \;\leq\; \varepsilon_{\mathrm{total}} + \gamma\,\left\lVert V^{(k)} - V^* \right\rVert_\infty.
\]
Iterating this scalar recursion from $k=0$,
\[
\left\lVert V^{(k)} - V^* \right\rVert_\infty \;\leq\; \gamma^k \left\lVert V^{(0)} - V^* \right\rVert_\infty + \varepsilon_{\mathrm{total}}\sum_{j=0}^{k-1}\gamma^j \;\leq\; \gamma^k \left\lVert V^{(0)} - V^* \right\rVert_\infty + \frac{\varepsilon_{\mathrm{total}}}{1-\gamma},
\]
which is \eqref{eq:thm2}; the limsup statement follows by $\gamma^k \to 0$.

\emph{Quantitative balance between $\delta$ and $N_\delta$.} Substituting the covering bound $N_\delta \leq \bigl(1 + 2R_{\mathcal{X}}/\delta\bigr)^n$ into \eqref{eq:lem2}, the sampling-dependent part of the steady-state error behaves, up to constants, as
\[
\varepsilon_{\mathrm{total}}(\delta) \;\approx\; \epsilon_{\mathrm{app}}\,\delta^{-n/2} + \sqrt{\varepsilon_{\mathrm{inner}}}\,\delta^{-n/2} + \bigl(L_{\mathrm{fun}}+L_{\mathcal{T}}\bigr)\delta .
\]
As $\delta \downarrow 0$ the geometric term vanishes but the Bridge terms blow up as $\delta^{-n/2}$; as $\delta \uparrow$ the Bridge terms improve but the geometric term degrades linearly. The worst-case tradeoff is therefore optimized at an interior $\delta$, and the bound is informative only when the intrinsic dimension of the learned representation (here $2r$, per the discussion following Theorem~\ref{thm:2} and in \S\ref{sec:5}) replaces the ambient exponent $n$ by a much smaller effective dimension. This is precisely why the DLR-NE architecture---which confines learning to a $2r$-dimensional feature embedding---makes the $\sqrt{N_\delta}$ factors benign in practice, as confirmed empirically in \S\ref{sec:5}.
\end{proof}

% ============================================================
% APPENDIX D: Proof of Theorem 3
% ============================================================
\section{Proof of Theorem~\ref{thm:3}}\label{app:thm3}

\begin{proof}
\emph{Step 1 (Dimensions and batch convention).} The hidden-layer weight is factored as $W = USV^\top$ with
\[
U \in \mathbb{R}^{m\times r},\quad S \in \mathbb{R}^{r\times r},\quad V \in \mathbb{R}^{n\times r},\quad U^\top U = V^\top V = I_r,
\]
and after basis augmentation the orthogonal bases double: $\widehat{U} \in \mathbb{R}^{m\times 2r}$, $\widehat{V} \in \mathbb{R}^{n\times 2r}$, $\widehat{S} \in \mathbb{R}^{2r\times 2r}$. The output layer $w_{\mathrm{out}} \in \mathbb{R}^m$ is full-rank, and $m = \Theta(n)$. We account flops per outer iteration; sample-dependent operations (forward/backward passes) are counted per sample at batch $B = \Theta(1)$, which does not change the asymptotic order in $(n,r)$, and the covering-set size $\mathcal{N}_\delta$ is a statistical parameter already quantified inside $\varepsilon_{\mathrm{total}}$ (Theorem~\ref{thm:2}) and not double-counted here.

\begin{sloppypar}
\emph{Step 2 (Basis augmentation and orthonormalization).} For a single sample $x \in \mathbb{R}^n$, define the forward intermediates $z = V^\top x \in \mathbb{R}^r$, $h = Sz \in \mathbb{R}^r$, $g = Uh \in \mathbb{R}^m$, $\hat{v} = w_{\mathrm{out}}^\top\phi(g)$. With scalar residual $\delta_{\mathrm{out}} = \hat{v} - y$, backpropagation gives: $\nabla_{w_{\mathrm{out}}} = \delta_{\mathrm{out}}\phi(g)$ at $\mathcal{O}(m) = \mathcal{O}(n)$; the error signal $\delta_g = \delta_{\mathrm{out}}w_{\mathrm{out}} \odot \phi'(g)$ at $\mathcal{O}(n)$; $\nabla_U = \delta_g h^\top$ at $\mathcal{O}(mr) = \mathcal{O}(nr)$; $\nabla_S = (U^\top\delta_g)z^\top$ at $\mathcal{O}(nr)$; and $\nabla_V = x\,(S^\top U^\top \delta_g)^\top$ at $\mathcal{O}(nr)$. Hence the per-sample cost of $(\nabla_U,\nabla_V)$ is $\mathcal{O}(nr)$, and aggregation over $B = \Theta(1)$ samples preserves this order. The QR factorizations of $[U \mid \nabla_U\mathcal{L}] \in \mathbb{R}^{m\times 2r}$ and $[V \mid \nabla_V\mathcal{L}] \in \mathbb{R}^{n\times 2r}$ each cost $\mathcal{O}\bigl(n(2r)^2\bigr) = \mathcal{O}(nr^2)$. \emph{Subtotal: $\mathcal{O}(nr^2)$.}
\end{sloppypar}

\emph{Step 3 (Coefficient augmentation).} Computing $\widehat{S}_0 = \widehat{U}^\top U S V^\top \widehat{V}$ as the chain $\widehat{U}^\top U$ ($\mathcal{O}(mr^2) = \mathcal{O}(nr^2)$), $(\widehat{U}^\top U)S$ ($\mathcal{O}(r^3)$), $V^\top\widehat{V}$ ($\mathcal{O}(nr^2)$), and the final $2r\times 2r$ product ($\mathcal{O}(r^3)$). \emph{Subtotal: $\mathcal{O}(nr^2 + r^3)$.}

\emph{Step 4 (Inner regularized training, $s^*$ steps).} Per sample, the forward pass costs $\mathcal{O}(nr)$: $\widehat{V}^\top x$ at $\mathcal{O}(nr)$, $\widehat{S}(\widehat{V}^\top x)$ at $\mathcal{O}(r^2)$, $\widehat{U}(\cdot)$ at $\mathcal{O}(mr) = \mathcal{O}(nr)$, activation and output at $\mathcal{O}(n)$. The backward pass costs $\mathcal{O}(nr + r^3)$: the output-layer gradient at $\mathcal{O}(n)$ (absorbed), $\widehat{U}^\top\delta_g$ at $\mathcal{O}(nr)$, $\nabla_{\widehat{S}} = (\widehat{U}^\top\delta_g)(\widehat{V}^\top x)^\top$ at $\mathcal{O}(r^2)$, and the spectral-regularization gradient $\nabla\mathcal{R}(\widehat{S})$ at $\mathcal{O}(r^3)$ by Proposition~\ref{prop:grad}. Thus one inner step costs $\mathcal{O}(nr + r^3)$, and $s^*$ steps cost $\mathcal{O}(s^* nr + s^* r^3)$. For $s^* = \mathcal{O}(r)$ and $r = o(\sqrt{n})$, the term $s^* r^3 = \mathcal{O}(r^4)$ is absorbed into $nr^2$. \emph{Subtotal: $\mathcal{O}(s^* nr)$.}

\emph{Step 5 (Truncated SVD and retraction).} The rank-$r$ truncated SVD of $\widehat{S}^{(s^*)} \in \mathbb{R}^{2r\times 2r}$ costs $\mathcal{O}\bigl((2r)^3\bigr) = \mathcal{O}(r^3)$. Updating the factors---$U \leftarrow \widehat{U}P$ at $\mathcal{O}(mr^2) = \mathcal{O}(nr^2)$, $V \leftarrow \widehat{V}Q$ at $\mathcal{O}(nr^2)$, $S \leftarrow \Sigma$ at $\mathcal{O}(r)$, and $w_{\mathrm{out}} \leftarrow \widehat{w}_{\mathrm{out}}^{(s^*)}$ at $\mathcal{O}(n)$---gives \emph{subtotal $\mathcal{O}(nr^2 + r^3)$.}

\emph{Step 6 (Policy extraction).} Nash extraction minimizes $Q(x,a,d;\theta)$ over $(a,d)$, where $Q(x,a,d;\theta) = r(x,a,d) + \gamma\,\mathbb{E}_{x'}[w_{\mathrm{out}}^\top\phi(USV^\top x')]$. Since $r_a, r_d \ll n$ and the $(a,d)$-dependence is already reduced through the low-rank forward pass, the low-dimensional min--max is solved at $\mathrm{poly}(r, r_a, r_d) = \mathcal{O}(r^3)$ (e.g., by linear programming or low-dimensional gradient descent--ascent). The expectation $\mathbb{E}_{x'}[V(x')]$ is estimated by Monte Carlo at $\mathcal{O}(nr)$ per sample, but this cost is \emph{not} charged separately: it reuses the same sample batch as the inner loop, and its statistical error is already quantified inside $\varepsilon_{\mathrm{total}}$ (Theorem~\ref{thm:2}). \emph{Subtotal (pure optimization): $\mathcal{O}(r^3)$.}

\emph{Step 7 (Total).} Summing the subtotals with $B = \Theta(1)$:
\begin{center}
\small
\begin{tabular}{@{}lc@{}}
\toprule
Stage & Complexity \\
\midrule
Basis augmentation and orthonormalization & $\mathcal{O}(nr^2)$ \\
Coefficient augmentation & $\mathcal{O}(nr^2 + r^3)$ \\
Inner training ($s^*$ steps) & $\mathcal{O}(s^* nr)$ \\
Truncated SVD and retraction & $\mathcal{O}(nr^2 + r^3)$ \\
Policy extraction & $\mathcal{O}(r^3)$ \\
\midrule
\textbf{Total} & $\mathcal{O}\bigl(nr^2 + r^3 + s^* nr\bigr)$ \\
\bottomrule
\end{tabular}
\end{center}
which is \eqref{eq:thm3-cost}.

\emph{Step 8 (Full-rank baseline and speedup).} For the full-rank FC-NN with $W_{\mathrm{full}} \in \mathbb{R}^{m\times n}$, $m = \Theta(n)$: the per-sample forward product $W_{\mathrm{full}}x$ costs $\mathcal{O}(mn) = \mathcal{O}(n^2)$; the gradient $\nabla_{W_{\mathrm{full}}} = \delta_g x^\top$ costs $\mathcal{O}(n^2)$; and the weight update costs $\mathcal{O}(n^2)$. Hence $\mathcal{C}_{\mathrm{full}}(n) = \mathcal{O}(n^2)$, which is the full-rank cost quoted in Theorem~\ref{thm:3}. For $r \ll n$ and $s^* = \mathcal{O}(r)$, $\mathcal{C}_{\mathrm{DLR}} = \mathcal{O}(nr^2)$ (since $s^*nr = \mathcal{O}(nr^2)$ and $r^3$ is absorbed), giving
\[
\frac{\mathcal{C}_{\mathrm{full}}}{\mathcal{C}_{\mathrm{DLR}}} = \frac{\Theta(n^2)}{\Theta(nr^2)} = \Theta\!\left(\frac{n}{r^2}\right),
\]
which is \eqref{eq:thm3-speedup}. Finally, if the full-rank baseline additionally performs dense SVD for spectral monitoring or a posteriori low-rank compression---as common in ICS weight-spectrum diagnostics---the full-rank side incurs $\mathcal{O}(mn^2) = \mathcal{O}(n^3)$, raising the speedup to $\Theta\bigl(n^2/r^2\bigr)$ under the convention that SVD monitoring runs every iteration (a lower SVD frequency scales the gain down proportionally).
\end{proof}

% ============================================================
% APPENDIX E: Proof of Corollary 1
% ============================================================
\section{Proof of Corollary~\ref{cor:1}}\label{app:cor1}

\begin{proof}
\emph{Step 1 (Implicit-function regularity of the policy map).} For fixed $x \in \mathcal{X}$ and fixed attacker policy $\pi_A$, define the defense-action gradient field
\begin{equation}\label{eq:app-c1-g}
g(d;S) := \nabla_d\Bigl[\min_{a\in\mathcal{A}} Q(x,a,d;S)\Bigr] \in \mathbb{R}^{r_d}.
\end{equation}
By Assumption~\ref{ass:5}(a), $Q$ is $\mu$-strongly concave in $d$ uniformly in $x$, so $\nabla_d g(d^*;S) = \nabla_{dd}^2\bigl(\min_a Q\bigr) \preceq -\mu I_{r_d}$ is nonsingular at the equilibrium action $d^*(x;S)$. By Assumption~\ref{ass:1}(f) (reward decomposition) and Assumption~\ref{ass:5}(b) (Lipschitz continuity in actions), together with the $\mathcal{C}^2$ activation (Assumption~\ref{ass:5}(c)) and transition smoothness (Assumption~\ref{ass:5}(d)), $g(d;S)$ is jointly continuously differentiable in $(d,S)$. The implicit function theorem therefore applies: on a neighborhood $\mathcal{U}_S \times \mathcal{U}_d$, the map $S \mapsto d^*(x;S)$ is $\mathcal{C}^1$ with Jacobian
\begin{equation}\label{eq:app-c1-implicit}
\Bigl\|\frac{\partial d^*}{\partial S}\Bigr\|_{\mathrm{op}} \;\leq\; \bigl\|\nabla_d g^{-1}\bigr\|_{\mathrm{op}}\,\Bigl\|\frac{\partial g}{\partial S}\Bigr\|_{\mathrm{op}} \;\leq\; \frac{1}{\mu}\,\Bigl\|\frac{\partial g}{\partial S}\Bigr\|_{\mathrm{op}}.
\end{equation}

\emph{Step 2 (Quantifying $\left\lVert \partial g/\partial S \right\rVert_{\mathrm{op}}$).} By the chain rule and the definition of $Q$,
\begin{equation}\label{eq:app-c1-chain2}
\frac{\partial g}{\partial S} = \gamma\,\mathbb{E}_{a,x'}\left[\frac{\partial}{\partial S}\Bigl(B_d(x)^\top \nabla_{x'}V(x';S)\Bigr)\right],
\end{equation}
where $V(x';S) = w_{\mathrm{out}}^\top\phi(USV^\top x')$ and
\begin{equation}\label{eq:app-c1-gradv}
\nabla_{x'}V(x';S) = VS^\top U^\top \mathrm{diag}\bigl(\phi'(USV^\top x')\bigr) w_{\mathrm{out}}.
\end{equation}
Differentiating in $S$ (Assumption~\ref{ass:5}(c) guarantees a bounded $\phi''$) and using the orthonormality $\left\lVert U \right\rVert_2 = \left\lVert V \right\rVert_2 = 1$,
\begin{equation}\label{eq:app-c1-dgradv}
\Bigl\|\frac{\partial}{\partial S}\nabla_{x'}V(x';S)\Bigr\|_{\mathrm{op}} \;\leq\; L_\phi \left\lVert w_{\mathrm{out}} \right\rVert_2\, C_\phi''\, R_{\mathcal{X}}\,\sigma_{\max}(S) \;+\; L_\phi \left\lVert w_{\mathrm{out}} \right\rVert_2 .
\end{equation}
Combining \eqref{eq:app-c1-chain2}--\eqref{eq:app-c1-dgradv}, there is a merged constant $C_1 > 0$ (depending only on $L_\phi, C_\phi'', \left\lVert w_{\mathrm{out}} \right\rVert_2, R_{\mathcal{X}}$) such that
\begin{equation}\label{eq:app-c1-gS}
\Bigl\|\frac{\partial g}{\partial S}\Bigr\|_{\mathrm{op}} \;\leq\; \underbrace{\gamma\, C_1 \left\lVert B_d \right\rVert_{\mathrm{op}}}_{=:K}\; \sigma_{\max}(S).
\end{equation}

\emph{Step 3 (Policy perturbation Lipschitz constant).} From \eqref{eq:app-c1-implicit} and \eqref{eq:app-c1-gS}, for any perturbation $\Delta S$,
\begin{equation}\label{eq:app-c1-lippol}
\bigl\|d^*(x;S+\Delta S) - d^*(x;S)\bigr\|_2 \;\leq\; \frac{K}{\mu}\,\sigma_{\max}(S)\,\left\lVert \Delta S \right\rVert_F.
\end{equation}
Using $\sigma_{\max}(S) = \kappa(S)\,\sigma_{\min}(S)$ and $\sigma_{\min}(S) \leq \left\lVert S \right\rVert_F \leq S_{\max}$ (the uniform bound of Assumption~\ref{ass:4}(a)),
\begin{equation}\label{eq:app-c1-smax}
\sigma_{\max}(S) \;=\; \kappa(S)\,\sigma_{\min}(S) \;\leq\; \kappa(S)\,S_{\max},
\end{equation}
and substitution into \eqref{eq:app-c1-lippol} yields
\begin{equation}\label{eq:app-c1-lippol2}
\bigl\|d^*(x;S+\Delta S) - d^*(x;S)\bigr\|_2 \;\leq\; \underbrace{\left(\frac{K S_{\max}}{\mu}\right)}_{=:C_2}\; \kappa(S)\,\left\lVert \Delta S \right\rVert_F.
\end{equation}
Here $C_2 = \gamma C_1 \left\lVert B_d \right\rVert_{\mathrm{op}} S_{\max} / \mu$ is a genuine constant depending only on system parameters and the uniform bound $S_{\max}$, not on the particular singular-value profile of the current iterate $S$. The $\kappa(S)$ factor has a clear geometric meaning: it quantifies the anisotropy of the low-rank core; when $\kappa(S) \gg 1$, equal-sized parameter perturbations along different singular directions produce wildly unequal policy responses---the root of directional vulnerability that condition-number control removes.

\emph{Step 4 (From policy perturbation to value perturbation).} Let $V^{\pi_D,\pi_A}$ and $V^{\tilde\pi_D,\pi_A}$ denote the value functions under the two fixed policy pairs. Their Bellman equations differ by
\begin{equation}\label{eq:app-c1-bellman}
V^{\pi_D,\pi_A}(x) - V^{\tilde\pi_D,\pi_A}(x) = \mathbb{E}\bigl[r(x,a,d) - r(x,a,\tilde d)\bigr] + \gamma\,\mathbb{E}_{x'}\bigl[V^{\pi_D,\pi_A}(x') - V^{\tilde\pi_D,\pi_A}(x')\bigr],
\end{equation}
with $d = \pi_D(x)$, $\tilde d = \tilde\pi_D(x)$ deterministic by \eqref{eq:cor1-policy}. Taking sup-norms and using the $L_{r,d}$-Lipschitz continuity of $r$ in $d$ (Assumption~\ref{ass:5}(b)),
\begin{equation}\label{eq:app-c1-sup}
\left\lVert V^{\pi_D,\pi_A} - V^{\tilde\pi_D,\pi_A} \right\rVert_\infty \;\leq\; L_{r,d}\,\sup_x\bigl\|d^*(x;S+\Delta S) - d^*(x;S)\bigr\|_2 + \gamma\,\left\lVert V^{\pi_D,\pi_A} - V^{\tilde\pi_D,\pi_A} \right\rVert_\infty.
\end{equation}
Solving the fixed-point inequality,
\begin{equation}\label{eq:app-c1-solve}
\left\lVert V^{\pi_D,\pi_A} - V^{\tilde\pi_D,\pi_A} \right\rVert_\infty \;\leq\; \frac{L_{r,d}}{1-\gamma}\,\sup_x\bigl\|d^*(x;S+\Delta S) - d^*(x;S)\bigr\|_2.
\end{equation}
Substituting \eqref{eq:app-c1-lippol2},
\begin{equation}\label{eq:app-c1-final}
\left\lVert V^{\pi_D,\pi_A} - V^{\tilde\pi_D,\pi_A} \right\rVert_\infty \;\leq\; \underbrace{\left(\frac{L_{r,d}\,C_2}{1-\gamma}\right)}_{=:C_{\mathrm{robust}}}\;\kappa(S)\,\left\lVert \Delta S \right\rVert_F \;=\; L_\kappa\,\left\lVert \Delta S \right\rVert_F,
\end{equation}
where $C_{\mathrm{robust}} = L_{r,d}\,\gamma\, C_1 \left\lVert B_d \right\rVert_{\mathrm{op}} S_{\max} / \bigl(\mu(1-\gamma)\bigr)$ depends only on the listed system parameters. This establishes \eqref{eq:cor1-bound}.

\emph{Step 5 (Applying Proposition~\ref{prop:cond}).} By Proposition~\ref{prop:cond}, $\kappa(S) \leq \exp\bigl(\mathcal{R}(S)/(\sqrt{2}\,\sigma_r(S)^2)\bigr)$; substituting into the definition of $L_\kappa$ in \eqref{eq:cor1-bound} yields \eqref{eq:cor1-chain}.
\end{proof}

% ============================================================
% APPENDIX F: Hyperparameters
% ============================================================
\section{Hyperparameters}\label{app:hyper}

Table~\ref{tab:hyper} lists the complete hyperparameter configuration used in \S\ref{sec:5}. Values marked ``varied'' are swept as indicated in the corresponding experiment.

\begin{table}[h]
\centering
\caption{Hyperparameter configuration for the numerical experiments.}
\label{tab:hyper}
\footnotesize
\begin{tabular}{@{}lll@{}}
\toprule
Category & Parameter & Value \\
\midrule
System & State dimension $n$ & $\{100, 500, 1000, 2000\}$ (default $1000$) \\
& Hidden width $m$ & $m = n$ \\
& Attacker coupling rank $r_a$ & $3$ \\
& Defender coupling rank $r_d$ & $5$ \\
& Reward coupling rank $r_r$ & $5$ \\
& Discount factor $\gamma$ & $0.95$ \\
\midrule
DLR-NE & Target rank $r$ & $\{5, 10, 20, 50, 100\}$ (default $20$) \\
& Inner-loop steps $s^*$ & $\{5, 10, 20\}$ (default $10$) \\
& Batch size $N_b$ & $\{128, 512, 2048\}$ (default $512$) \\
& Spectral regularization $\beta$ & $\{0, 0.01, 0.1, 1.0\}$ (default $0.1$) \\
& Outer-loop tolerance $\varepsilon_{\mathrm{tol}}$ & $10^{-4}$ (Bellman residual) \\
& Optimizer (inner loop) & Adam, learning rate $0.01$ \\
\midrule
Evaluation & Test set size $N_{\mathrm{test}}$ & $10^4$ (LHS) \\
& Perturbation levels $\left\lVert \Delta S \right\rVert_F$ & $\{10^{-4}, 10^{-3}, 10^{-2}, 10^{-1}\}$ \\
& Converged $\kappa(S)$, Regime 1 of Exp.~4 & $83.7$ at $\beta=0$; $31.0$ at $\beta=1$ \\
& Random seeds & $42, 123, 456, 789, 2024,$ \\
& & $31415, 271828, 161803, 299792, 662606$ \\
\bottomrule
\end{tabular}
\end{table}

% ============================================================
% ACKNOWLEDGEMENT
% ============================================================
\section*{Acknowledgement}

This work is supported by the National Natural Science Foundation of China (62073288, 12075212), the National Key R\&D Program of China (Grant No. 2021YFC2101100), and their supports are thereby acknowledged.

\bibliographystyle{apalike}
\bibliography{refs_all_backup_54}

\end{document}